\RequirePackage{snapshot} 
\documentclass[12pt,journal,draftclsnofoot,onecolumn,a4paper]{IEEEtran}

\usepackage{bm}
\usepackage{amssymb}
\usepackage{amsfonts}
\usepackage[cmex10]{amsmath}%
\usepackage{dsfont}
\usepackage[usenames,dvipsnames]{color}
\usepackage{fancyvrb}

\newcommand{\exclude}[1]{}

\usepackage{times}
\usepackage{amssymb}
\usepackage{amsfonts}
\usepackage[cmex10]{amsmath}
\usepackage{dsfont}
\usepackage{url}
\usepackage{subfigure}

\ifCLASSINFOpdf
\usepackage[pdftex]{graphicx}
\usepackage[pdftex,colorlinks,
           bookmarks=true,%
           pdftitle=Quality\ of\ Real-Time\ Streaming\ in\ Wireless\ Cellular\ Networks,%
           pdfauthor=B.\ Blaszczyszyn%
           \ -\ M.Jovanovic\ -\  M.K.\ Karray,]{hyperref} 
\hypersetup{
   bookmarksnumbered,
   pdfstartview={FitH},
   citecolor={blue},
   linkcolor={red},
   urlcolor=[rgb]{0,0.55,0},
   pdfpagemode={UseOutlines}
}
\usepackage[numbers,sort&compress]{natbib}
\usepackage{hypernat}
\else

 \usepackage[dvips]{graphicx}
\usepackage[numbers,sort&compress]{natbib}
\fi
\graphicspath{{./}{./figures/}}
\DeclareGraphicsExtensions{.pdf}

\newcommand{\bmX}{\bm{X}}
\newcommand{\bmY}{\bm{Y}}
\newcommand{\bmx}{\bm{x}} 
\newcommand{\bmr}{\bm{r}} 
 
\newcommand{\ind}{\mathds{1}}
\newcommand{\E}{\mathbf{E}}
\renewcommand{\Pr}{\mathbf{P}}
\newcommand{\calF}{\mathcal{F}}

\newcommand{\calL}{\mathcal{L}}
\newcommand{\calR}{\mathcal{R}}
\newcommand{\bbN}{\mathbb{N}}

\newcommand{\md}{\text{d}}

\newtheorem{Th}{Theorem}[section]
\newtheorem{prop}[Th]{Proposition}
\newenvironment{Prop}{\bf\begin{prop}\rm\em}{\end{prop}} 
\newtheorem{cor}[Th]{Corollary}
\newtheorem{res}[Th]{Result}
\newtheorem{lemma}[Th]{Lemma}
\newenvironment{Lemma}{\bf\begin{lemma}\rm\em}{\end{lemma}} 
\newtheorem{fact}[Th]{Fact}

\newtheorem{exe}[Th]{Example}
\newtheorem{remark}[Th]{Remark}
\newenvironment{Remark}{\bf\begin{remark}\rm}{\end{remark}} 

\begin{document}
\title{Quality of Real-Time Streaming
in Wireless\\[-1ex] Cellular Networks\\
\Large Stochastic Modeling and Analysis}

\author{Bart{\l }omiej~B{\l }aszczyszyn, Miodrag Jovanovic \\and  Mohamed Kadhem Karray,~\IEEEmembership{Member,~IEEE,}
\thanks{B. B{\l}aszczyszyn is with Inria-ENS, 
23 Avenue d'Italie, 75214 Paris, France; email: Bartek.Blaszczyszyn@ens.fr}
\thanks{M. Jovanovic and M.~K.~Karray are with  Orange Labs, 38/40 rue
  G\'{e}n\'{e}ral Leclerc,  92794  Issy-les-Moulineaux, France; email: 
$\{$miodrag.jovanovic, mohamed.karray$\}$@orange.com}
\thanks{This paper reports the results of the research undertaken  
under CRE-CIFRE thesis co-advising agreement  between Inria  and
Orange Labs.}}


\maketitle
\thispagestyle{empty}

\begin{abstract}
We present a new stochastic service model with capacity sharing and interruptions, appropriate for the evaluation of the quality of real-time streaming (e.g. mobile TV) in wireless cellular networks.  It takes into account multi-class Markovian process of call arrivals (to capture different radio channel conditions, requested streaming bit-rates and call-durations) and allows for a general resource allocation policy saying which users are temporarily denied the requested fixed streaming bit-rates (put in outage) due to resource constraints.  We develop general expressions for the performance characteristics of this model, including the mean outage duration and the mean number of outage incidents for a typical user of a given class, involving only the steady-state of the traffic demand.  We propose also a natural class of least-effort-served-first resource allocation policies, which cope with optimality and fairness issues known in wireless networks, and whose performance metrics can be easily calculated using Fourier analysis of Poisson variables.  We specify and use our model to analyze the quality of real time streaming in 3GPP Long Term Evolution (LTE) cellular networks. Our results can be used for the dimensioning of these networks.
\end{abstract}

\begin{IEEEkeywords}
Real-time streaming, stochastic model, mobile TV,  LTE, quality of service,
interruptions, outage, deep outage, 
capacity-sharing, Poisson process
\end{IEEEkeywords}

\maketitle

\section{Introduction}
\label{s.Intro}
Wireless cellular networks offer nowadays 
possibility to watch TV on mobile devices, which is an example of
a real-time content streaming.
This type of traffic demand is expected to increase
significantly in the future.  In order to cope with this process, 
network operators need to implement in their dimensioning tools
efficient methods allowing to predict the quality of this type of service.
The quality of real-time streaming (RTS) is principally related
to the number and duration 
of {\em outage incidents} ---  (hopefully short)
periods when the network cannot deliver 
to a given user in real-time the requested content of the required quality.
In this paper we propose a stochastic model
allowing for an analytic evaluation of such  metrics. 
It assumes a traffic demand with different radio
conditions of calls,
and can  be specified to take into account the parameters
of a given wireless cellular technology.
We develop expressions for several important performance characteristics 
of this model, including the  mean time 
spent in outage and the  mean number of outage incidents  for a typical
streaming call in function of its  radio conditions.
These expressions involve
only stationary probabilities of the (free) traffic demand process,
which is a  vector of independent Poisson random
variables describing the number of users in different radio conditions.

We use this model 
to analyze RTS in a typical cell of a 3GPP Long Term Evolution
(LTE)  cellular network assuming 
orthogonal intra-cell user channels  with the peak bit-rates 
(achievable when there are no other users in the same cell)
close to the theoretical  Shannon's bound in the additive
white Gaussian noise (AWGN) channel, with the extra-cell interference
treated as noise.
These  assumptions lead to a radio resource constraint 
in a multi-rate linear form.
Namely, each user experiencing  a given
signal-to-(extra-cell)-interference-and-noise  ratio  (SINR) 
requires a fixed fraction of the 
normalized radio capacity,
related to the ratio between its  requested and peak bit-rates.
All users of a given configuration (experiencing 
different SINR values) can be entirely satisfied if and only if
the total required capacity is not larger than one.
\footnote{Recall that in the case of voice calls and, more generally, constant
bit-rate (CBR) calls 
the multi-rate linear form of the resource constraints 
has already proved to lead to efficient model evaluation methods, 
via e.g. Kaufman-Roberts algorithm~\cite{Kaufman1981,Roberts1981}.
Despite some fundamental similarities to CBR service,
the RTS gives rise to a new model, due to the fact that  the service denials 
are not definitive for a given
call, but have a form of temporal  interruptions (outage) periods.}

In the above context of a multi-rate linear radio resource constraint,
we  analyse some natural parametric class of {\em
least-effort-served-first} (LESF) service policies, which  assign service 
to users in order of their increasing radio capacity demand, until the full 
capacity (possibly with some margin) is reached. The capacity margin
may be used to offer some ``lower quality'' service to users
temporarily in outage
thus realizing  some type of fairness with respect to 
unequal user radio-channel conditions.
This class contains an optimal and a fair
policy, the latter being suggested by  LTE implementations.

In order to evaluate explicitly the quality of service metrics induced by the  LESF
policies we relate the 
mean time spent
in outage and the mean number of outage incidents for a typical streaming call in given 
radio conditions
to the distribution functions of some linear functionals
of the Poisson vector describing   the steady state of the system.
We calculate the Fourier transforms of these functions   
and use  a well-known  Fourier transform 
inversion  method 
to obtain numerical values of the quantities of interest.
We also study the mean throughput  during a typical streaming call 
evaluating the expectations of the corresponding non-linear functionals of the 
Poisson vector  describing   the steady state of the system via the Monte Carlo method.

Using this approach, we  present a thorough study of the quality of 
RTS with LESF policies in the aforementioned Markovian setting. 
For completeness we present also some pure-simulation results
illustrating  the impact of a non Poisson-arrival assumption.

\bigskip

Let us now recollect a few {\em related works} on the performance
evaluation of cellular networks.
In early 80's, wireless cellular networks were carrying
essentially voice calls, which  require 
constant bit-rates (CBR) and are  subject to 
admission control policies with blocking (at the arrival epoch)  to
guarantee these rates for calls already in
service. An important amount of
work has been done to propose efficient call 
admission policies~\cite{Zander1992Dist,Yates1995AFramework,Sampath1995}. 
Policies with admission conditions in the multi-rate linear form have been
considered
e.g. in~\cite{BaccelliBlaszczyszynTournois2003,ElayoubiHaddadaFourestie2008,Karray2010Analytical}.

Progressively, cellular networks started carrying  also calls with variable
bit-rates (VBR), used to transmit data files. The
available resources are (fairly) shared between such calls
and when the traffic demand increases, the file transfer delays 
increase as well,  but (in principle) no call is ever blocked.
These delays may be evaluated analytically
using multi-rate linear resource constraint 
in conjunction with  multi-class processor sharing
models; cf e.g.~\cite{BonaldProutiere2003,Karray2010Analytical}.

Recently, users may access multimedia streaming services through their
mobile devices~\cite{Fitzek2004}. They are provided via
CBR connections, essentially without admission control,
but they tolerate temporary interruptions, 
when network congestions occur.
One  may distinguish two types of
streaming traffic. In {\em real-time streaming (RTS)} (as e.g. in  mobile
TV), considered in this paper, the portions of  the streaming content
emitted during the time when the transmission to a  given user is 
interrupted (is in outage) are definitely lost for him (unless a
``secondary'', lower-rate streaming is provided during these periods).
In {\em non-real-time streaming (NRTS)} (like e.g., video-on-demand,
YouTube, Dailymotion, etc), a user starts playing back
the requested multimedia content after some initial delay, required to deliver 
and buffer on the user device 
some initial portion of  it. If further transmission is
interrupted for some time making the user buffer content drop to zero
(buffer starvation) 
then the play-back is stopped until some new required portion of the
content is delivered.  
Several papers study the effect of the variability
of the wireless channel on the performance of a single streaming call; see for
e.g.~\cite{Liang2008}, \cite{Medard2010}. 
In~\cite{Elayoubi11} VBR transmissions and RTS are considered
jointly in some analytical model, however the number and duration of outage
periods are not evaluated.
In \cite{Altman-etal11} the tradeoff between the start-up
delay and the probability of buffer starvation is analyzed 
in a Markovian queuing framework for NRTS streaming.

We do not consider any cell-load balancing; see~\cite{bethanabhotla2013joint} for some recent work on this problem in the video streaming context.
Also, \cite{li1998analysis,wu2002dynamic} consider some 
admission control policies to guarantee
non-dropping of multimedia calls due to  caller impatience and/or  handoffs.



{\em The remaining part of this paper} is organized as follows.
In Section~\ref{s.wireless} we will present our model for the 
evaluation of the quality of RTS  in wireless
cellular networks.
Technical proofs of the results presented in this section are postponed
to the Appendix, where they are given in a more general context.
Section~\ref{s.numerical} 
specifies our model to be compliant with the LTE cellular networks
specification and  presents numerical
results regarding the quality of RTS in these networks.

\section{Streaming in wireless cellular networks}
\label{s.wireless}
In this section we  present a new stochastic model  
of RTS in cellular networks. 

\subsection{System  assumptions}
We consider the following scenario of multi-user streaming in 
a cellular network.
\subsubsection{Network layer}
Geographically distributed users 
wish to obtain down-link wireless streaming of some (typically video)
content, contacting base stations of a network at  random times, for
random durations, 
requesting some fixed streaming bit-rates. 
We consider a uni-cast traffic (as opposed to the broadcast or
multi-cast case),  i.e.; the content is delivered to all
users via private connections.
Different classes of users
(calls) need to be distinguished, 
regarding their radio channel conditions,
requested  streaming bit-rates and mean streaming times.
Each user chooses one base station, the one  with the smallest
  path-loss,  independently of the
  configuration of users served by this station.
Thus, we do not consider any load-balancing policy.

\subsubsection{Data layer --- streaming policies}
If a given base station cannot serve all the users present 
at a given time, it temporarily stops streaming the requested 
content at the requested rate to users of some classes, according to
some given policy  (to be described), which is supposed to  preserve a
maximal subset of served users. We call these (classes of) users with the requested bit-rate
temporarily  denied {\em in outage}. The users in outage 
will not receive the part of the content which is emitted 
during their outage times (this is the principle of the RTS).
We will also consider policies,
which offer some ``best-effort'' streaming bit-rates for some classes
of users in outage, thus allowing for example to keep receiving  the
requested content but of  a lower  quality. Users, which are
(temporarily) denied even this lower quality of service 
are called {\em in deep
  outage}.
\subsubsection{Medium access } In this paper we assume that users are
connected to the serving antennas 
via  orthogonal single-input-single-output (SISO) channels allowing for 
the peak-rate close to the theoretical  Shannon's bound in the additive
white Gaussian noise (AWGN) model, with the (extra-cell) interference treated as noise.%
\footnote{Orthogonality  of channels  is an appropriate assumption 
for current LTE (Long Term Evolution) norm for cellular networks based
on OFDMA, as well as for other multiple access techniques as FDA, TDMA, CDMA assuming perfect in-cell orthogonality, and even HDR neglecting the
scheduler gain.} We will also comment on how to model 
multiple-input-multiple-output (MIMO) and broadcast channels. 

\subsubsection{Physical layer} 
\label{sss.PhysicalL}
The quality of channel of a given user depends on  the
path-loss of the signal with respect to its serving base station,
a constant noise, and the interference from other (non-serving)
base stations. These three components determine its
signal-to-interference-and-noise ratio (SINR).
Both path-loss form the serving station and interference
account for the distance and random propagation effects (shadowing).
Our main motivation for considering  a multi-class model is 
to distinguish users with different SINR values.
In other words, even if we assume that all users require the same 
streaming times and  rates,
we  still  need a multi-class model due to
(typically) different SINR's values of users in 
wireless cellular networks.

\subsubsection{Performance characteristics}
We will present and analytically evaluate performance of 
some (realistic) streaming policies in the context described above. 
We will be particularly interested in the following characteristics:
\begin{itemize}
\item  fraction of  time spent in outage and in deep outage during the typical
call  of a given class,
\item number of outage incidents occurring during this call, 
\item mean throughput (average bit-rate) during such call, accounting
  for the requested bit-rates and for the ``best-effort'' bit-rate
  obtained during the outage periods.
\end{itemize}

\subsection{Model description} 
In what follows we describe a mathematical model
of the RTS that is an incarnation of a new, more general, 
stochastic service model with capacity sharing and interruptions presented and analyzed
in the Appendix~\ref{s.model}. This is a single server model which allows to study  
the performance of one tagged base station  of a multi-cellular network satisfying the above system  assumptions.
More details on how this model fits the multi-cell scenario will be presented  in Section~\ref{s.numerical}.

\subsubsection{Traffic demand}

\label{s.UsersLocations-Erlang}
Consider $J\ge 1$ {\em classes of calls} (or, equivalently,  users)
characterized  by different
{\em requested streaming bit-rates} $r_k$, {\em wireless channel
  conditions} described by the signal-to-(extra-cell)-interference-and-noise ratio
$\text{SINR}_k$ with respect to the serving base-station~\footnote{In this paper
the interference is always  caused only by non-serving base stations.}
and  {\em mean requested streaming times} $1/\mu_k$, $k=1,\ldots,J$.

We assume that calls of class
$k\in\{1,\ldots,J\}$ arrive in time according
to a Poisson process 
with intensity $\lambda_{k}>0$ (number of call arrivals per unit of time, per
base station) 
and stay in the system (keep requesting streaming) for independent  
times, having some {\em general distribution} with mean 
$1/\mu_{k}<\infty$.~\footnote{All the results presented in this paper 
do not depend on the particular choice of the streaming time
distributions. This property is often referred to in the queuing
context as the insensitivity property.}  
Different classes of calls are independent from each other.
We denote by $X_k(t)$ the number of calls of a given class 
requesting streaming from a given BS at time $t$;
see Section~\ref{s.UsersLocations} in the Appendix for a formal definitions of
these variables in terms of arrival process and service times.
Let $\bmX(t)=(X_1(t),\ldots,X_J(t))$; we call it the (vector of) user
configuration at time $t$. The stationary distribution $\pi$ of
$\bmX(t)$ coincides with  
the distribution of the vector $(X_1,\dots,X_J)$
of independent Poisson random variables
with means $\E[X_k]:=\rho_{k}=\lambda_{k}/\mu_{k}$, 
$k=1,2,\ldots,J$.
We call $\rho_k$ the  \emph{traffic demand} (per base station) of class $k$.

\subsubsection{Wireless resource constraints}
\label{ss.ressource-policy}
Users are supposed to be offered the requested streaming rates
for the whole requested streaming times. 
However, due to limited wireless  resources,  for some
configuration of users $\bmX(t)$, the requested streaming rates
$\bmr=(r_1,\ldots,r_J)$
may be not achievable. 
Following the assumption of orthogonal AWGN SISO
wireless channels (with the (extra-cell) interference treated as noise) 
available for
users of a given station, we assume that 
the requested rates  are achievable for all calls
present at time $t$ if  
\begin{equation}\label{e.orthogonality}
X_k(t)r_k=\nu_k r^{\max}_k, \quad k=1,\ldots,J,
\end{equation}
for some  non-negative
vector $(\nu_1,\ldots,\nu_J)$, such that
$\sum_{k=1}^J\nu_k\le 1$, where 
\begin{equation}\label{e.phi-snr}
r^{\max}_k=\gamma W\log(1+\text{SINR}_k)
\end{equation}
is the maximal (peak) bit-rate of a user of class
$k$, whose channel conditions are characterized by  $\text{SINR}_k$.
(The rate $r^{\max}_k$ is available to a user of class $k$
if it is the only user served by the base station.)
Here  $W$ is the
frequency bandwidth and  $\gamma$ (with $0<\gamma\le 1$) 
is a coefficient telling how close a given coding scheme approaches
the theoretical Shannon's  bound 
(corresponding to $\gamma=1$);
cf~\cite[Th~.9.1.1]{CoverThomas2006}.  
\footnote{It was also shown in~\cite{KJ2012}  that the performance of 
AWGN {\em multiple input multiple output} (MIMO) channel can be
approximated by taking values of $\gamma\ge1$. 
Another possibility to consider MIMO channel 
is to use the exact capacity  formula  given
in~\cite{T99}.}
Note that the assumption~(\ref{e.orthogonality})
corresponds to the situation, 
when  users neither hamper nor assist  each other's transmission. 
They use channels which are perfectly separated in time, frequency 
or by orthogonal codes, nevertheless sharing these resources.
\footnote{From information theory point of view, the orthogonality assumption  
is not optimal. In fact, the theoretically optimal performance is
offered by the {\em broadcast channel} model.
It is known that in the case of AWGN broadcast channel
the  rates $\bmr$ are (theoretically) achievable for the
configuration $\bmX$ 
if (and only if) there exists a vector 
$(\nu_1,\ldots,\nu_J)$, such that $\sum_{k=1}^J\nu_k\le 1$
and  
$$X_kr_k=W\log\biggl(1+\frac{\nu_k}{1/\text{SINR}_k
+\sum_{i=1}^{k-1}\nu_i}\biggr)\,\quad k=1,\ldots,J,$$
where the classes of users are numbered such that
$\text{SINR}_1\ge \text{SINR}_2\ge\ldots\ge \text{SINR}_J$;
cf~\cite[Eq.~6.29]{TseViswanath2005}. 
}

We can interpret the ratio between the requested and maximal bit-rates
 $\varphi_k=r_k/r^{\max}_k$ 
as the {\em resource demand} of a user of class $k$.
Note that the configuration of users $\bmX(t)$ can be entirely served 
if and only if the total resource demand satisfies the constraint
\begin{equation}\label{e.MErlang}
\sum_{k=1}^J\varphi_k
X_k(t)\le 1\,.
\end{equation}
This is a {\em multi-rate linear resource constraint}.

\subsubsection{Service policy}
\label{sss.Service-policy}
If the requested streaming rates are not achievable for 
a given configuration of users $\bmX(t)$ present at time $t$, 
then some classes of users  will be temporarily 
put in outage at time $t$, 
meaning that they will receive some smaller bit-rates (whose values
are not guaranteed and may depend on the configuration   $\bmX(t)$).
These smaller, ``best-effort''  bit-rates may drop to 0, in which case we
say that users  are in deep-outage.
Let us recall that the times at which users are in outage
and deep outage do {\em not} alter the original streaming times;
i.e. the streaming   
content is not buffered, nor delayed during the outage periods.

We will define now  a parametric family of service
polices for which  
{\em classes with smaller resource demands  have higher
service priority}. 
In this regard, in the remaining part of the paper
we assume (without  loss of generality) that the
resource demands of users from different classes 
are  ordered $\varphi_1<\varphi_2<\ldots<\varphi_J$.

\paragraph{Least-effort-served-first policy}
For a given configuration of users $\bmX=\bmX(t)$ requesting streaming
at time $t$,  {\em least-effort-served-first policy with $\delta
$-margin} ({\em LESF($\delta$)} for short)
attributes the requested bit-rates to all users in classes
$k=1,\ldots,K$, where 
\begin{equation}\label{e.LESF}
K=K^\delta(\bmX)=
\max\left\{k\in\{1,\ldots,J\}:
\sum_{j=1}^{k-1}\varphi_{j}X_j+\varphi_k\sum_{j=k}^JX_j\ind(\varphi_j\le \varphi_k(1+\delta))%
\leq1\right\}\,,
\end{equation}
where $\ind_A(x)=1$ is the indicator function of set $A$ 
and  $\delta$ is a constant satisfying $0\le\delta\le\infty$.

\begin{Remark}\label{r.LESF}
The LESF$(0)$ policy is  {\em optimal} in the following sense:
given  constraint~(\ref{e.MErlang})  and the  assumption 
that the  classes with smaller resource demands  have higher
priority, this policy allows to serve the maximal subset of users
present in the system. For the same reason any  LESF$(\delta)$ policy with
$\delta>0$ is clearly sub-optimal. 
In order to explain the motivation  for considering such policies, 
one needs to extend the model and 
explain what actually happens with  
classes of users which experience outage.
In this regard, note that
$C=\sum_{j=1}^{K} \varphi_j X_j\le 1$ is the
actual fraction of the server capacity consumed by the users
which are not in outage. 
The remaining server capacity $1-C$ (which is not needed to serve 
users in classes $1,\ldots,K$) can be used to offer some 
``lower quality'' service (e.g. streaming with lower video resolution, etc)
to the users in classes $K+1,\ldots, J$ which are in outage. Note
by~(\ref{e.LESF}) 
that  the remaining server capacity under the policy LESF$(\delta)$ 
is at least
$$1-C\ge \varphi_{K} 
\sum_{j=K+1}^J X_j\ind(\varphi_j\le \varphi_K(1+\delta))\,.$$
Hence, the server accepting the class  $K$ as the 
least-priority class being ``fully'' served,  
leaves enough remaining  capacity to be able to make the same effort
(allocate service capacity $\varphi_{K}$) 
for all users in outage in classes  whose service
demand exceeds $\varphi_K$
by no more than $\delta\times100\%$. These latter users will not have 
``full'' required service (since
this requires more resources, $\varphi_j>\varphi_K$, for the full service)
but only  some ``lower quality'' service (to be specified in what
follows).  Consequently, 
one can conclude that  policies LESF$(\delta)$  with
$\delta>0$, being sub-optimal,  ensure some {\em fairness}, in
the sense explained above. Clearly the policy  LESF$(\infty)$
(i.e., with $\delta=\infty$) is the most fair,
in the sense that it reserves enough remaining capacity
to offers the ``lower quality'' service 
for {\em all} users in outage (no deep outage). Thus, we will call
LESF$(\infty)$ the {\em LESF fair} policy.
\end{Remark}

\paragraph{Best-effort service for users in outage}
We will specify now a natural
model for the ``best-effort'' streaming bit-rates that can be
offered for users in outage in association with a given  LESF($\delta$) policy.
For $k> K=K^\delta(\bmX)$ denote
\begin{equation}\label{e.sub-rate}
r_k'=r_k^{'\delta}(\bmX)=r_k^{\max}
\frac{1-\sum_{j=1}^{K}X_j\varphi_j}%
{\sum_{j=K+1}^JX_j\ind(\varphi_j\le (1+\delta)\varphi_K)}\quad
\text{if $\varphi_k\le (1+\delta)\varphi_K$ and 0 otherwise.} 
\end{equation}
The rates $(r_1,\ldots,r_K,r'_{K+1},\ldots,r'_J)$ are achievable
for the configuration $\bmX$
under resource constraint~(\ref{e.MErlang}).
Note that users in classes $j$ such that $\varphi_j>
(1+\delta)\varphi_K$ do not receive any positive bit-rate. We say,
they are in {\em deep outage}. Finally, we remark that
the service~(\ref{e.sub-rate}) is ``resource fair''
among users in outage but not in deep outage.

\subsubsection{Performance metrics}
\label{sss.Perf-eval}
Configuration of users $X(t)$ evolves in time,
it changes at arrival and departure times of users.
At each arrival or departure epoch the base station applies the outage
policy to the new configuration of users to decide which classes of
users receive requested streaming rates and which are in outage (or
deep outage). 

Let us introduce the following characteristics of the {\em typical
 call (user)} of class $k=1,\ldots,J$.
\begin{itemize}
\item  $P_k$ denotes the {\em probability of outage at the arrival
    epoch for class $k$}. This is the probability that the typical call
of this class  is put in outage  immediately at its arrival epoch.
\item $D_k$ denotes the {\em mean total time spent in 
outage during the typical call of class  $k$}.
\item $M_k$ denotes the {\em mean  number of outage incidents 
experienced during the typical call of class $k$}. 
\end{itemize}
More  formal definitions of these characteristics, 
as well as other {\em system} characteristics (as e.g. the intensity of
outage incidents) is  given in the Appendix.
We also introduce two further characteristics related to the mean
{\em throughput} obtained during the typical call of class $k=1,\ldots,J$. 
\begin{itemize}
\item Denote by $T_k$ the {\em  mean throughput during the typical
  call of class $k$}. This is the mean 
bit-rate obtained during such a call,
taking into account the bit-rate $r_k$ when the call is not in outage
and the best-effort bit rate $r_k'$ obtained during the outage
periods, averaged over call duration. 
\item Let  $T'_k$  be the {\em part of the throughput obtained
during the outage periods of the typical call of class $k$}. This is 
 the mean best-effort  bit-rate of such call averaged over
outage periods.
\end{itemize}

\subsection{Model evaluation}
\label{ss.results-Erlang}
\subsubsection{Results}
We will show how the performance metrics regarding outage incidents
and duration, introduced in Section~\ref{sss.Perf-eval},
can be expressed using probability distribution functions of some
{\em linear} functionals of the  random vector
$X_1,\ldots,X_J$ of independent  Poisson random variables
with parameters $\rho_j$, respectively. Recall that these random
variables correspond to the number of calls of different classes
present in the stationary regime of our streaming model.

Specifically, 
for given $\delta>0$, $k=1,\ldots,J$ and $t\ge0$ denote 
\begin{equation}\label{e.FkDelta}
F_k^\delta(t):=
\Pr\left\{\sum_{j=1}^{k}X^{\delta,k}_{j}\varphi_{j}\leq t\right\}\,,
\end{equation}
where $X_j^{\delta,k}=X_j$ for $j=1,\ldots,k-1$ and
$X^{\delta,k}_k=\sum_{j=k}^JX_j\ind(\varphi_j\le \varphi_k(1+\delta))$.

The following results follow from the analysis of a more general model
presented in the Appendix.
\begin{Prop}\label{p.Erlang}
The probability of outage at the arrival epoch for  user
of class~$k$
is equal to 
\begin{equation}
P_k=1-F_k^\delta(1-\varphi_k)\, \qquad k=1,\ldots,J\,.
\label{e.InterruptionProbability-Erlang}%
\end{equation}
The mean total time spent in outage during the typical call of class
$k$ is equal to 
\begin{equation}\label{e.InterruptionDuration-Erlang}%
D_k=\frac{P_k}{\mu_{k}}= \frac{1-F_k^\delta(1-\varphi_k)}{\mu_k}\,\qquad k=1,\ldots,J\,.
\end{equation}
 The mean  number of outage incidents experienced during the typical
 call of class $k$ (after its arrival) is equal to 
\begin{equation}
M_k=\frac{1}{\mu_{k}}\sum_{j=1}^{J}\lambda_{j}
\Bigl(F_k^\delta(1-\varphi_k)-F_k^\delta(1-\varphi_k-\varphi_j)\Bigr)
\qquad k=1,\ldots,J\,.  \label{e.NbInterruptions-Erlang}%
\end{equation}
\end{Prop}

\begin{proof}
Note first that the  functions $F^\delta_k(t)$ defined in~(\ref{e.FkDelta})
 allow one to represent the stationary probability  that the 
configuration of users is in a  state in which 
the LESF($\delta$) policy serves users of class $k$
$$F^\delta_k(1)=\Pr\left\{\sum_{j=1}^{k}X^{\delta,k}_{j}\varphi_{j}\le 1
\right\}\,.$$ 
In the general model described in the  Appendix
we denote this state by $\calF_k$ and its probability by $\pi(\calF_k)$. Thus $\pi(\calF_k)=F^\delta_k(1)$.  Moreover, 
$$1-F_k^\delta(1-\varphi_k)=\Pr\left\{\sum_{j=1}^{k}X^{\delta,k}_{j}\varphi_{j}>
1-\varphi_k\right\}$$ 
is the probability that  the  steady state configuration of users appended with one user of class $k$ 
is in the complement $\calF'_k$ of the state $\calF_k$,
i.e., all users of class $k$ are in outage (meaning  $k>K^\delta(\bmX')$, where
$\bmX'=(X_1,\ldots,X_k+1,\ldots,X_J)$). 
Thus the expression~(\ref{e.InterruptionProbability-Erlang}) 
follows from Proposition~\ref{p.InterruptionProbability}.
Similarly~(\ref{e.InterruptionDuration-Erlang})
follows from Proposition~\ref{p.InterruptionDuration}
and~(\ref{e.NbInterruptions-Erlang})
follows from Proposition~\ref{p.NbInterruptions}.
\end{proof}


Regarding the throughput characteristics, we have the following result. 
\begin{Prop}\label{e.Throughput-Erlang}
The mean throughput during the typical call of class $k$ is equal to 
$$T_k=r_k(1-P_k)+T_k'=r_k F_k^\delta(1-\varphi_k)+T_k'
\,,$$ 
where
\begin{equation}\label{eq.Tprim}
T_k'=\E\left[r_k^{'\delta}(X_1,\ldots,X_k+1,\ldots,X_J)
\ind\Bigl(K^\delta(X_1,\ldots,X_k+1,\ldots,X_J)<k\Bigr)\right]\,,
\end{equation}
with the best-effort rate $r'_k(\cdot)$ given by~(\ref{e.sub-rate}) 
and the  least-priority class $K^\delta(\cdot)$ begin served by the
$LESF(\delta)$ policy given by~(\ref{e.LESF}), is the  
part of the throughput obtained during the outage periods.
\end{Prop}
{\em Proof} of this proposition is given in the Appendix.

\begin{Remark}\label{r.SINR-outage-region}
Recall from~(\ref{e.sub-rate}) that the variable rates $r'_k$
are obtained by the user of class  $k$ when he is in
outage, i.e., $k>K$. They are non-null, $r'_k>0$,  only if
 $\varphi_k\le (1+\delta)\varphi_K$.
In the case of equal requested rates $r_k$, the intersection of the
two conditions
$0<r'_k$ and $k>K$  is equivalent to
\begin{equation}\label{eq.outageSINR}
(1+\text{SINR}_K)^{1/(1+\delta)}-1\le
\text{SINR}_k\le\text{SINR}_K\,.
\end{equation}
\end{Remark}

\subsubsection{Remarks on numerical evaluation}
\label{sss.Laplace}
In order to be able to use  the expressions given in~(\ref{p.Erlang})
we need to  evaluate the distribution functions  $F^\delta_k(t)$.
In what  follows we show how this can be done using Laplace
transforms. Regarding  the throughput in outage  $T_k'$, 
expressed in~(\ref{eq.Tprim}) as the expectation of a {\em non-linear}
functional of the vector  $(X_1,\ldots,X_J)$,
we will use  Monte Carlo simulations to obtain numerical
values for this expectation.

Denote by 
$\calL^\delta_k(\theta):=\int_0^\infty e^{-\theta s}F^\delta_k(s)\md
s$ the Laplace transform of  the function $F^\delta_k(t)$.
\begin{fact}
We have 
\[
\calL^\delta_k(\theta)=\frac{1}{\theta}\exp\left[
\sum_{j=1}^{k}\rho^{\delta,k}_{j}\left(  e^{-\theta\varphi_{j}}-1\right)  \right]\,,
\]
where $\rho^{\delta,k}_j=\rho_j$ for $j=1,\ldots,k-1$ and 
$\rho^{\delta,k}_k=\sum_{j=k}^J\rho_j\ind(\varphi_j\le \varphi_k(1+\delta))$.
\end{fact}
\begin{proof}
Note that for given $\delta>0$, $k=1,\ldots,J$ the random variables
$X^{\delta,k}_1,\ldots,X^{\delta,k}_k$ are  independent, of   Poisson distribution,
with parameters $\rho^{\delta,k}_1,\ldots,\rho^{\delta,k}_k$, respectively.
The result follows from~\cite[Proposition
  1.2.2]{BaccelliBlaszczyszyn2009T1} and a general relation $\int_0^\infty e^{-\theta s}F(s)\,\md s=\frac1\theta\int_0^\infty e^{-\theta s}F(\md s)$.
\end{proof}

The probabilities $F^\delta_k(\cdot)$  may be
retrieved from $\calL^\delta_k(\cdot)$ using standard techniques.
For example~\citep[with the  algorithm implemented by
Hollenbeck~\cite{Hollenbeck1998} in Matlab]{HoogKnightStokes1982}.
In what follows we present a more explicit result 
based on the Bromwich contour inversion integral.
In this regard, denote
$\overline\calL^\delta_k(\theta)=1/\theta-\calL^\delta_k(\theta)$
(which is the Laplace transform of complementary distribution function
$1-F^\delta_k(t)$). Also, denote by $\calR(z)$ the real part of the complex number~$z$.
\begin{fact}
We have
\begin{equation}\label{eq:Bromwich}
F^\delta_k(t)=1-\frac{2e^{at}}{\pi}
\int_{0}^\infty \calR\left(\overline\calL^\delta_{k}
(a+iu)\right)\cos ut\,\md u\,,
\end{equation}
where $a>0$ is an arbitrary constant.
\end{fact}
\begin{proof}
See~\cite{AbateWhitt95}.
\end{proof}

\begin{Remark}
As shown in~\cite{AbateWhitt95}, the integral in~(\ref{eq:Bromwich}) 
can be numerically evaluated using the trapezoidal rule, with  the
parameter~$a$ allowing to control the approximation error. Specifically, 
for $n=0,1,\ldots$ define
$$
h_n(t)=h_n(t;a,k,\delta):=\frac{(-1)^ne^{a/2}}{t}
\calR\left(\overline\calL^\delta_{k}
\Bigl(\frac{a+2n\pi i}{2t}\Bigr)\right)\,,
$$
$S_n(t):=\frac{h_0(t)}{2}+\sum_{i=1}^n h_i(t)$,  and 
$S(t)=\lim_{n\to\infty} S_n(t)$.
Then  $\left|F_k^\delta(t)-(1-S(t))\right|\le e^{-a}$.
Finally, the (alternating) infinite series $S(t)$ can be efficiently approximated  
using for example the Euler summation rule 
$$S(t)\approx \sum_{i=0}^M{M \choose i}\,2^{-M}S_{N+i}(t)\,$$
with a typical choice $N=15$, $M=11$.
\end{Remark}

\begin{Remark}\label{r.KR}
The expression~(\ref{e.NbInterruptions-Erlang}) for the mean  number of outage incidents 
involves a sum of potentially big number of  terms $F_k^\delta(1-\varphi_k)-F_k^\delta(1-\varphi_k-\varphi_j)$, $j=1,\ldots,J$,
which are typically small, and which are evaluated via the inversion of the Laplace transform.
Consequently the sum may accumulate precision errors. In order to avoid this problem we propose another numerical 
approach for calculating $M_k$. It consists in representing   $M_k$ equivalently to (\ref{e.NbInterruptions-Erlang}) as
\begin{equation}\label{e.NbInterruptions-Erlang-KR}%
M_k=\frac{F_k^\delta(1-\varphi_k)}{\mu_{k}}\sum_{j=1}^{J}\lambda_{j}b_{k}(j)
\qquad k=1,\ldots,J\,.
\end{equation}
where 
%
\begin{equation}
b_{k}\left(  j\right)  =\frac{F_{k}^{\delta}\left(  1-\varphi_{k}\right)
-F_{k}^{\delta}\left(  1-\varphi_{k}-\varphi_{j}\right)  }{F_{k}^{\delta
}\left(  1-\varphi_{k}\right)  }\label{e.bjk}%
\end{equation}
Let $k$ and $\delta$\ be fixed. Recall the definition of  $F_{k}^{\delta}\left(  t\right)  $
in ~(\ref{e.FkDelta}) and note that the expression~(\ref{e.bjk})
may be written as
\[
b_{k}\left(  j\right)  =\frac{\Pr\left(  X\in\mathcal{F},X+\epsilon_{j}%
\notin\mathcal{F}\right)  }{\Pr\left(  X\in\mathcal{F}\right)  }%
\]
where $\mathcal{F}=\mathcal{F}(k)=\left\{  X\in\mathbb{R}^{J}:\sum_{j=1}^{k}X_{j}^{\delta
,k}\varphi_{j}\leq1-\varphi_{k}\right\}  $. The above expression may be seen
as the blocking probability for class $j$ in a classical multi-class  Erlang loss system
with the admission condition $X\in\mathcal{F}$. Consequently, 
 $b_{k}\left(  \cdot\right)  $ may be calculated by using
the \emph{Kaufman-Roberts algorithm}~\cite{Kaufman1981,Roberts1981}
and plugged into~(\ref{e.NbInterruptions-Erlang-KR}).
Note that by doing this we still need to calculate 
$F_k^\delta(1-\varphi_k)$ however avoid summing of $J$ differences of
these functions as in~(\ref{e.NbInterruptions-Erlang}). 
%
\end{Remark}

\section{Quality of real-time streaming in LTE} 
\label{s.numerical}
In this section we will use the model developed in Section~\ref{s.wireless}
to evaluate the quality of RTS in LTE  networks.
This single-server (base station) model will be used  to study the performance of one tagged base station  of a multi-cellular network under the following assumptions: 
\begin{itemize}
\item  We assume a regular hexagonal lattice of base stations on a torus. This allows us to consider the tagged  base station of the network as a typical one. 
\item Homogeneous (in space and time) Poisson arrivals on the torus are marked by  i.i.d. (across users and base stations) variables representing their shadowing with respect to different base stations. These variables, together with independent user locations determine  their  serving (strongest) base stations.  
A consequence of the independence of users locations and shadowing variables is that
the arrivals served by the tagged base station form an  independent thinning of the total Poisson arrival process to the torus and thus a Poisson  process too. Uniform distribution of user locations and identical distribution of the their shadowing variables imply that the intensity of the arrival process to the tagged base station
is equal to the total arrival intensity  to the torus divided by the number of stations.   
Moreover, the distribution of the SINR of the typical user of the tagged base station coincides with the distribution of the typical user of the whole network.     
\item The intensity of arrivals of some particular (SINR)-class to the tagged base station
is equal to the total intensity of arrivals to the tagged cell times the probability of  the random SINR of the typical user being in the SINR-interval corresponding to this class. 
\item We consider the ``full interference'' scenario, i.e., that all base stations emit the signal with the constant power, regardless of the number of users they serve (this number can be zero). This makes the interference, and hence the  service rates, of users of a given base station
independent of the service of other base stations (decouples the service processes of different base stations).
\end{itemize}

\vspace{-2ex}
\subsection{LTE model and traffic  specification}
\subsubsection{SINR distribution}
\label{ss.SINR}
Recall that the main motivation for considering a  multi-class model
was the necessity to  distinguish users with different radio 
conditions, related to  different values of the SINR they have with
respect to the serving base stations. 
In order to choose representative values of SINR in a given
network and to know what  fraction of users experience a given value,
we need to know the {\em (spatial) distribution 
of the SINR} (with respect to the serving base station)
experienced in this network (possibly biased by the spatial 
repartition of arrivals of streaming calls). 
This distribution can be obtained  from real-network measurements, 
simulations or  analytic
evaluation of an appropriate spatial, stochastic model.%
\footnote{For this latter
possibility, we refer the reader to a recent paper on 
Poisson modeling of real cellular networks subject to
shadowing~\cite{hextopoi}, as well as to~\cite{DHILLON2012},
completed in~\cite{coverage}, where the distribution of the 
the SINR in Poisson networks is evaluated explicitly.}
In this paper we will use the distribution 
of SINR obtained from the simulation compliant with the 3GPP recommendation 
in the so-called calibration case (to be explained in what follows).
At present, assume simply, that we are given a cumulative distribution
function (CDF) of the SINR expressed in dB,
$F(x):=\Pr\{10\log_{10}(\text{SINR})\le x\}$, obtained from either of these
methods. In other words, $F(x)$  represents  the fraction of mobile users
in the given network which experience the SINR (expressed in dB) not larger than $x$.

Consider a discrete probability mass function 
\begin{equation}\label{e.pk}
p_k:=F\Bigl(\frac{x_{k+1}+x_{k}}2\Bigr)-
F\Bigl(\frac{x_{k}+x_{k-1}}2\Bigr)\qquad k=1,2,\ldots,J\,,
\end{equation}
with $x_0=-\infty$, $x_{J+1}=\infty$. We define the class $k=1,\ldots,J$
of users as all users having  the SINR expressed in dB
in the interval $\bigl((x_{k}+x_{k-1})/2,
(x_{k+1}+x_{k})/2\bigr]$, and approximate their SINR by the common
value $\text{SINR}_k=10^{x_k/10}$. Clearly $p_k$  is  the fraction
of mobile users in the given network which experience the SINR
close to  $\text{SINR}_k$. Hence, in the case of a homogeneous streaming
traffic (the same requested streaming rates and mean streaming times,
which will be our default assumption in the numerical examples) 
we  can assume the intensity of arrivals $\lambda_k$ of users of class
$k$ to be equal to $\lambda_k=p_k\lambda$
where $\lambda=\sum_{i=k}^J\lambda_k$ is the total 
arrival intensity (per unit of time per serving base station) to be specified together with the CDF $F$ of the SINR.
 

\paragraph{CDF of the SINR  for  3GPP recommendation} 
\label{ss.SINR-LTE}
We obtain the CDF $F$ of SINR 
from the simulation compliant with the 3GPP recommendation 
in the so-called calibration case, (compare to~\cite[Figure~A.2.2-1%
(right)]{3GPP-36.814}). 
More precisely, 
we consider the geometric pattern of  BS 
placed on the $6\times6$ hexagonal lattice.
In the middle of each hexagon there are three symmetrically oriented
BS antennas, which gives a total of 108 BS antennas. The distance between the
centers of two neighboring hexagons is $0.5\,$km. 
Each BS antenna is characterised by the following  horizontal 
pattern  $A(\phi)=-\min(12(\phi/\theta)^2,A_m)$, where $\phi$
is the angle in degrees, with  $\theta = 70^\circ$, $A_m = 20dB$,
and uses transmission power 
$P = 60\text{dBm}$  (including omnidirectional gain of $14\text{dBi}$).
The distance-loss model (corresponding to the frequency carrier 
2GHz) is $L(r)= 128.1 + 37.6 × \log_{10}(r)[\text{dB}]$
where $r$ is the distance in km. 
A supplementary penetration loss of 20dB is added. 
The shadowing is modeled as a  log-normal random variable of
mean one and logarithmic standard deviation  of
deviation~8dB, cf~\cite{propag_extabst}.
The noise power equals
$-95\text{dBm}$  (which corresponds to a system bandwidth of 10MHz, a
noise floor of -174dBm/Hz and a noise figure of 9dB).
In order to obtain the empirical CDF of the SINR we  generate
 3600 random user locations uniformly in the network
(100 user locations per hexagon on average).
Each user is connected to the antenna  with the strongest received
signal (smallest propagation-loss including distance, shadowing and antenna
pattern) and the SINR is calculated.
The obtained empirical CDF~$F$  of the SINR is shown on Figure~\ref{f.cdf}.

\begin{figure}[t!]
\begin{center}
\includegraphics[width=0.7\linewidth]{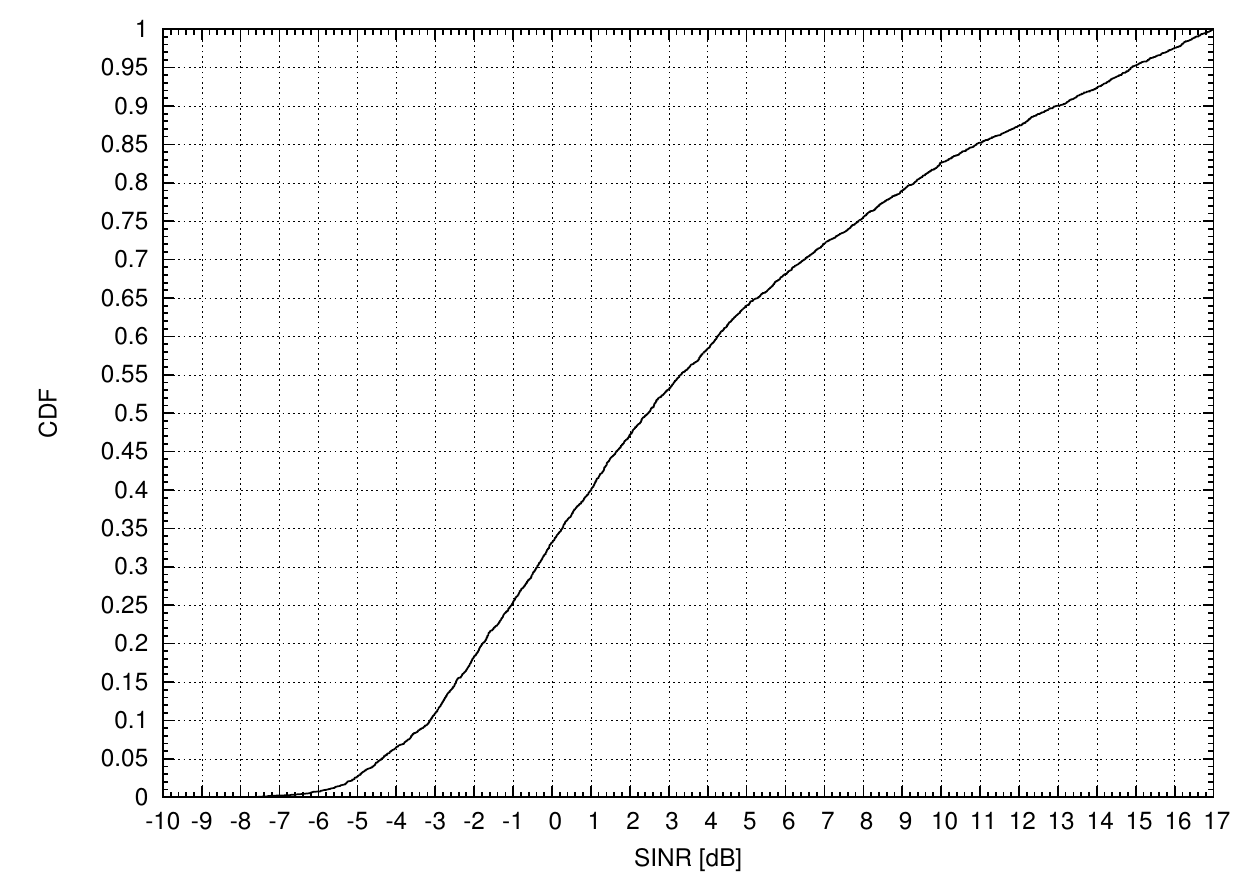}
\end{center}
\caption{Cumulative distribution function  of the SINR obtained according to 3GPP specification;
see Section~\ref{ss.SINR-LTE}. An abrupt transition of the CDF to 1 at $\text{SINR}=17\text{dB}$ is due to the 
cell sectorization: each mobile is  interfered by each of  the two antennas  co-located with its serving antenna on the same site 
(and serving the different sectors) with the power  equal  to at least 
$1\%$ of the power  received from the serving BS. Therefore the signal to  interference 
ratio is at most  $0.5\times 10^{-2}=−17\text{dB}$. 
\label{f.cdf}}
\end{figure}

\subsubsection{Link characteristics}
3GPP shows in~\cite[\S A.2]{3GPP-36.942}  that
there is a 25\% gap between the practical coding schemes and
the Shannon's limit for the AWGN channel.
Moreover, some of the transmitted bits are used for signaling,  which
induces a supplementary capacity loss of about 30\%
(see~\cite[\S6.8]{3GPP-36.211}). This made us assume $\gamma
=0.5(\approx 0.75(1 - 0.3))$ in~(\ref{e.phi-snr}).
The system bandwidth is $W=10\text{MHz}$.

\subsubsection{Streaming traffic}
We assume that all calls require  the same  streaming rate $r_k=
256\,\text{kbit/s}$ and have the same streaming call  time distribution.
We  split them into  $J=100$ user classes characterized by  values of
the SINR falling into different intervals
regularly approximating the SINR domain from 
 $x_1=-10\text{dB}$ to $x_J=17\text{dB}$
as explained in Section~\ref{ss.SINR}.
In our performance evaluation we will consider two values of the
spatially uniform traffic demand: 900 and 600 Erlang/km${}^2$.
(All results presented in what follows  do not depend on the mean
streaming time but only on the traffic demand).
Consequently, $k\,$th class traffic demand per unit of surface
is equal to,   respectively,  
$p_k\times900$ and $p_k\times600\,\text{Erlang}/\text{km}^2$,
where $p_k$ are given by~(\ref{e.pk}). Multiplying  by the surface
served by one base station equal to $\sqrt{3}\cdot
(0.5\,\text{km})^2/6\approx 0.0722\,\text{km}^2$
we obtain the traffic demand per cell, per class, equal to 
$\rho_k=p_k\times 900\times 0.0722\approx p_k\times 64.9$~Erlang
and $\rho_k=p_k\times 600\times 0.0722\approx p_k\times 43.3$~Erlang,
respectively, for the two studied scenario.

\subsection{Performance evaluation}
\label{ss.Nres}
Assuming the LTE and traffic model described above, we consider now
streaming policies LESF($\delta$) defined in Section~\ref{sss.Service-policy}.
Recall  that in doing so, we assume that users are served by the antenna
offering the smallest path-loss, and dispose orthogonal down-link
channels, with the maximal rates
$r_k^{\max}$ depending on the value
of the SINR (interference comes from non-serving BS) 
characterizing  class $k$.  Roughly speaking, 
LESF($\delta$) policy 
assigns the total requested streaming rate $r_k=256\text{kbit/s}$ for
the maximal possible subset of classes in the order of decreasing
SINR, leaving some capacity margin to offer some
``best-effort''  streaming rates for (some) users remaining in outage.  
These streaming rates  $r'_k$ given by~(\ref{e.sub-rate})
depend on the current configuration of users and 
are non-zero for  
users with SINR within the interval
$(1+\text{SINR}_K)^{1/(1+\delta)}-1\le
\text{SINR}\le\text{SINR}_K$,
where $\text{SINR}_K$ is the minimal 
value of SINR for which users are assigned the total requested
streaming rate; cf Remark~\ref{r.SINR-outage-region}.
In particular, 
LESF($0$), called the {\em optimal} policy, leaves no capacity margin for
users in outage, while  LESF($\infty$), called the {\em fair} one,
offers a  ``best-effort'' streaming rate for all users in outage
at the price of assigning the full requested rate $256$kbit/s to a smaller
number of classes (higher value of SINR$_K$)~\footnote{The LESF fair
  policy seems to be adopted  in  some implementations of the LTE.}.
In what follows, we use our results of 
Section~\ref{ss.results-Erlang}  to evaluate
performance of these streaming policies in the  LTE network model.

\subsubsection{Outage time}
\label{sss.outage-time}
Figure~\ref{f.duration900} shows
the mean time of the streaming call spent in 
outage normalized by call duration, $\mu_kD_k$, evaluated using~(\ref{e.InterruptionDuration-Erlang}),
in function of the SINR value
characterizing class $k$, for the  traffic 900~Erlang/km${^2}$
and different policies LESF$(\delta)$. Figure~\ref{f.duration600} 
shows the analogous results assuming traffic load of 600~Erlang/km${^2}$. 
The main observations are as follows:
\begin{itemize}
\item All  LESF policies exhibit a cut-off behaviour: the fraction of
  time in outage drops rapidly from 100\% to 0\% when SINR transgresses
  some critical values. This cut-off is more strict for the optimal policy.
 \item For the traffic of  900~Erlang/km${^2}$,  users with  SINR$\ge3\text{dB}$ are practically never in outage,
when the optimal policy is used. The same holds true for users with
SINR$\ge13\text{dB}$, when the fair policy is used.
\item When the traffic drops to 600~Erlang/km${^2}$, these critical
  values of SINR decrease by $2\text{dB}$ and $5\text{dB}$,
  respectively, for the optimal and the fair policy. Note that the
  fair policy is more sensitive to higher traffic load.
\end{itemize}

\begin{figure}[t!]
\begin{center}
\includegraphics[width=0.7\linewidth]{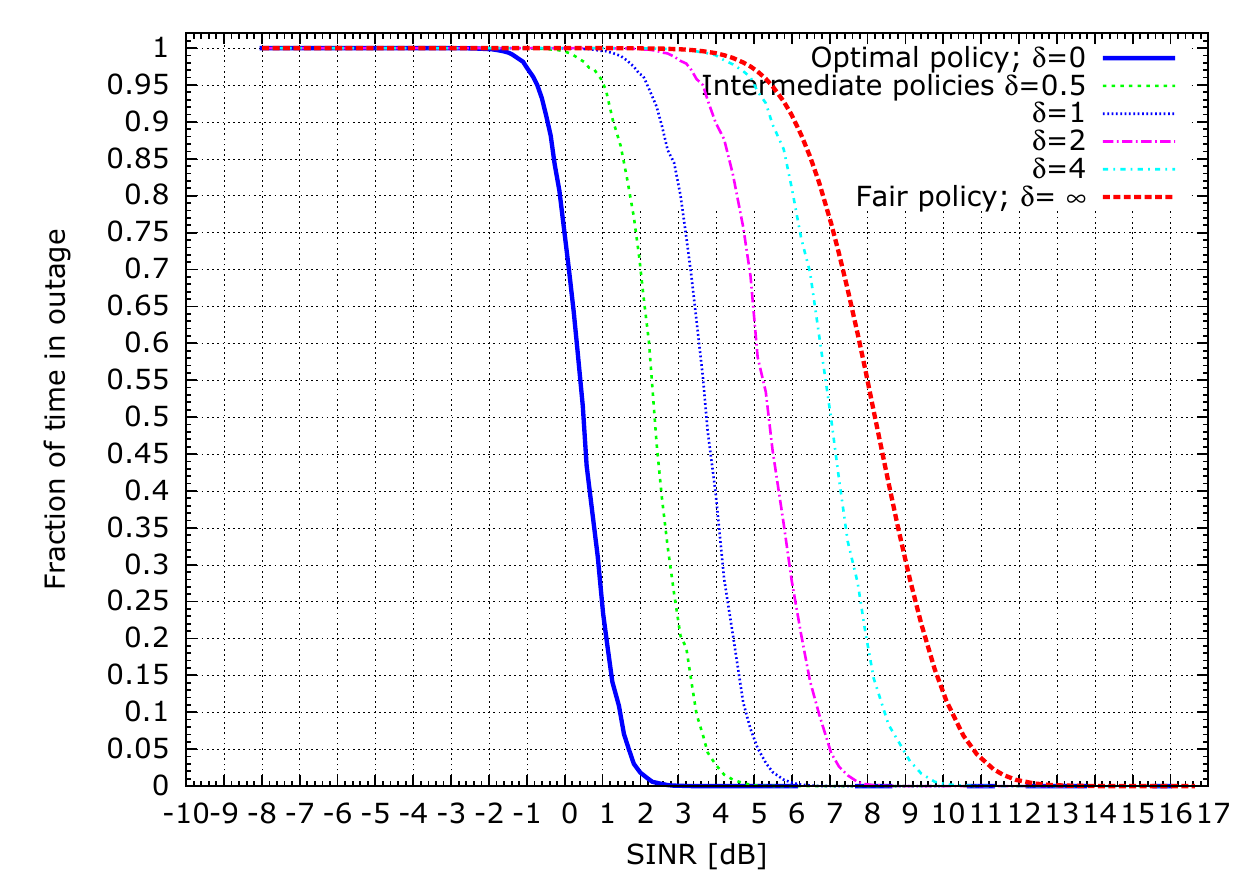}
\end{center}
\vspace{-4ex}
\caption{Mean fraction of the requested streaming time in 
  outage,  in function of the user SINR 
for different policies LESF$(\delta)$; traffic 900~Erlang/km${^2}$. 
\label{f.duration900}}
\begin{center}
\includegraphics[width=0.7\linewidth]{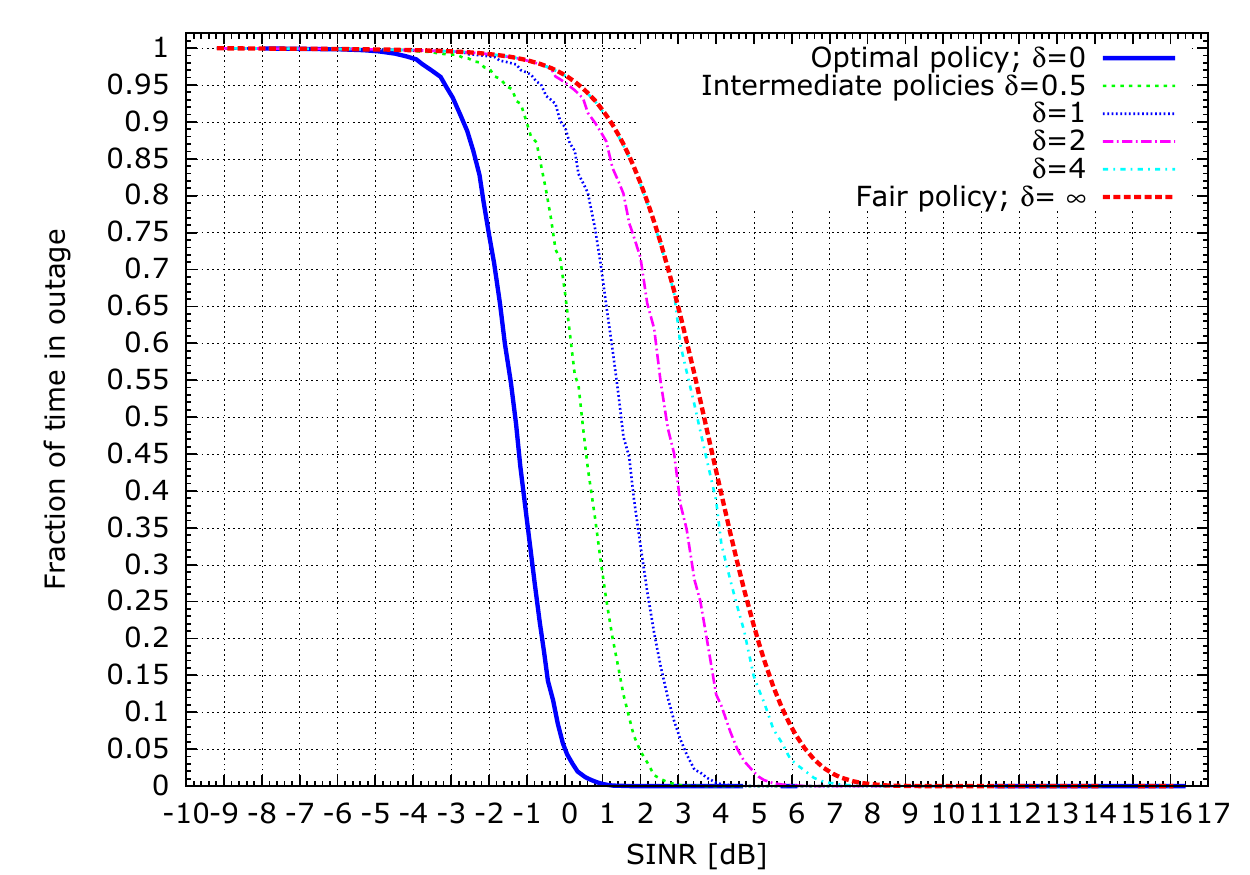}
\end{center}
\vspace{-4ex}
\caption{Fraction of time in outage as on Figure~\ref{f.duration900}
for traffic 600~Erlang/km${^2}$. 
\label{f.duration600}}
\vspace{-4ex}
\end{figure}

\subsubsection{Number of outage incidents}
\label{sss.outage-incidents}
Figure~\ref{f.number900} shows the mean number of outage incidents
per streaming call, $M_k$  evaluated using~(\ref{e.NbInterruptions-Erlang}),
in function of the SINR value characterizing class $k$, 
for the  traffic 900~Erlang/km${^2}$
and different policies LESF$(\delta)$. 
(Recall that  we assume the same streaming time distribution for all
users, and hence $\lambda_j/\mu_k=\rho_j$ making  the expression
in~(\ref{e.NbInterruptions-Erlang}) depend only on the vector of traffic demand
per class.)   
Figure~\ref{f.number600} 
shows the analogous results assuming traffic demand of 600~Erlang/km${^2}$. 
The main observations are as follows:
\begin{itemize}
\item For all policies, the number of outage incidents (during the
  service) is non-zero 
only for users with the SINR close to the critical values 
revealed by the analysis of
the outage times. Users with SINR below these values are constantly in outage
while users with SINR above them never in outage. 
\item More fair policies generate slightly more outage incidents.
The worst values are 2 to 2.2 interruptions per call for the
optimal policy, depending on the traffic value, and  2.4 to 3  interruptions per
call for the fair  policy.
\end{itemize}
Studying outage times and outage incidents we do not see 
apparent reasons for considering fair policies. This motivates our
study of the best-effort service in outage.

\begin{figure}[t!]
\begin{center}
\includegraphics[width=0.7\linewidth]{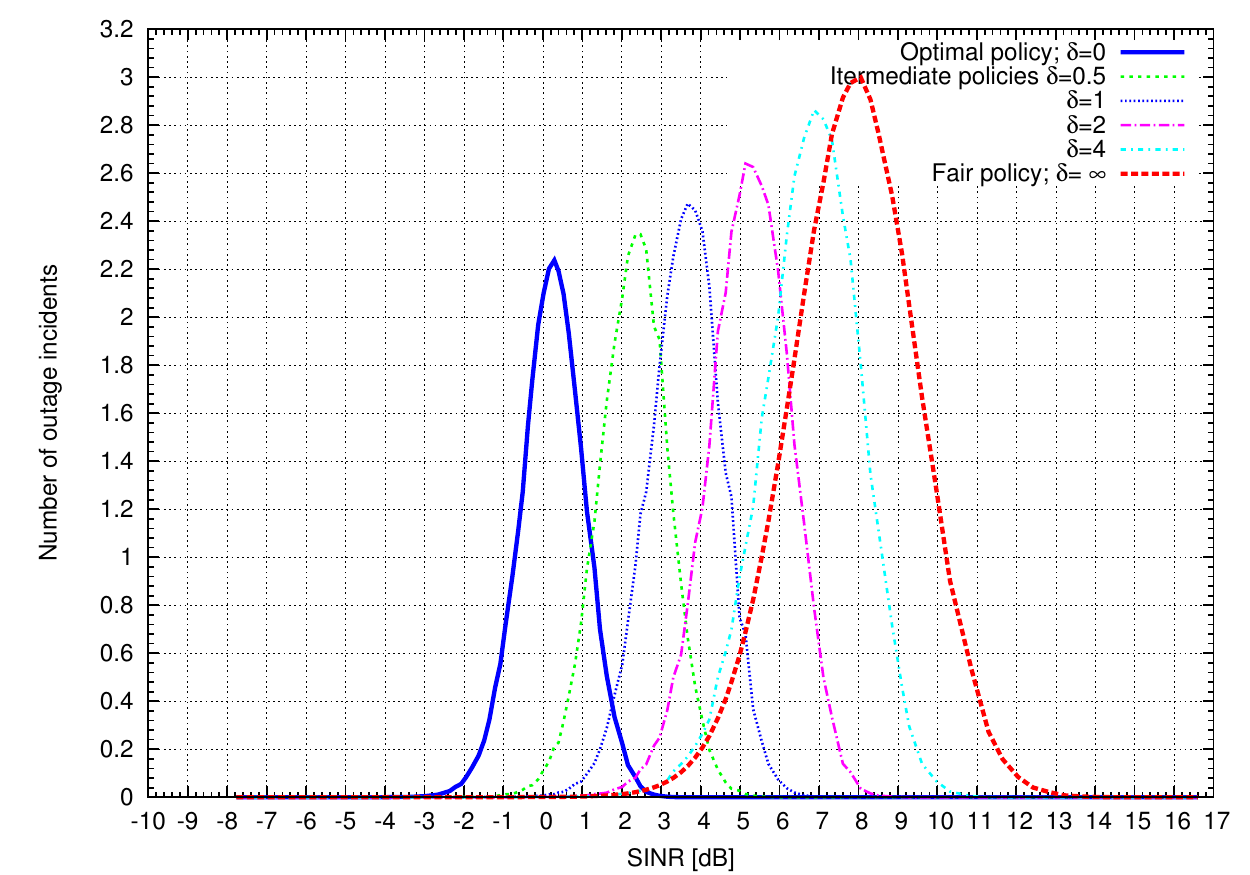}
\end{center}
\vspace{-4ex}
\caption{Number of outage incidents during the 
requested streaming time,  in function of the user SINR 
for different policies LESF$(\delta)$; traffic 900~Erlang/km${^2}$. 
\label{f.number900}}
\begin{center}
\includegraphics[width=0.7\linewidth]{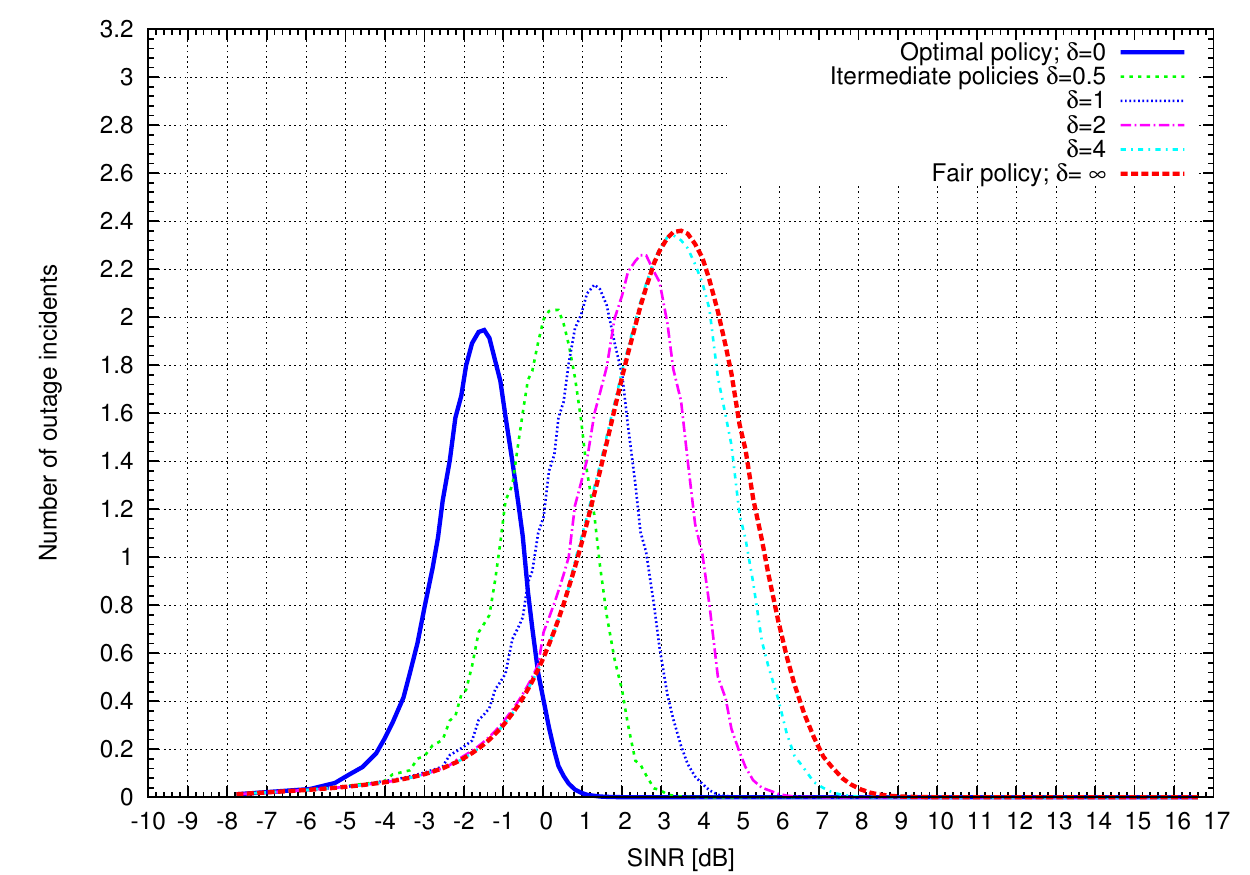}
\end{center}
\vspace{-4ex}
\caption{Number of outage incidents as on Figure~\ref{f.number900}
for traffic 600~Erlang/km${^2}$. 
\label{f.number600}}
\vspace{-4ex}
\end{figure}

\begin{figure}[t!]
\begin{center}
\includegraphics[width=0.7\linewidth]{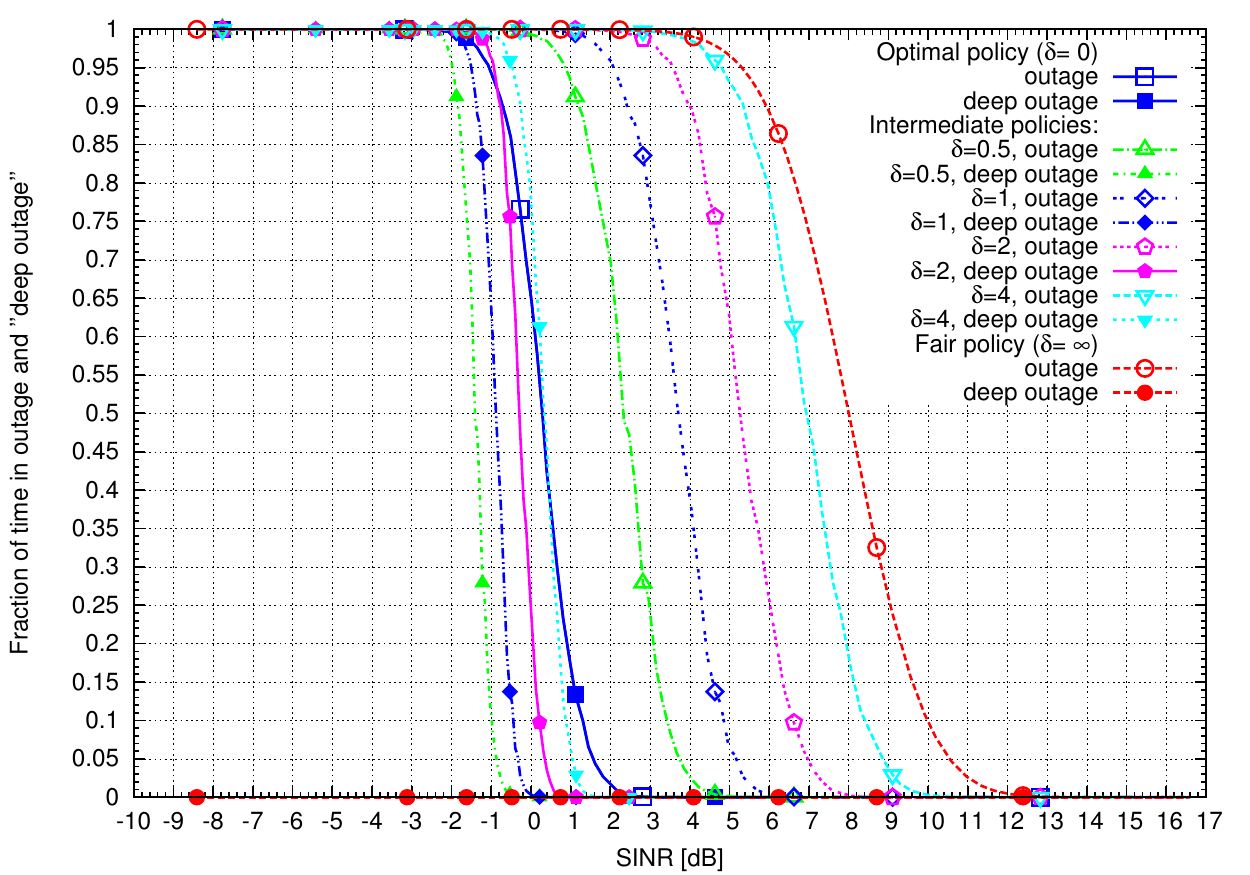}
\end{center}
\vspace{-4ex}
\caption{Deep outage versus outage time.
For any policy
  LESF$(\delta)$, with $0<\delta<\infty$, the left curve of a given
  style represents 
 the fraction of time spent in deep outage.
The right curve of a given style recalls the fraction of time spent in outage
(already plotted on Figure~\ref{f.duration900}).
The optimal policy ($\delta=0$) does not offer any 
 ``best effort'' service. The fair policy ($\delta=\infty$) offers this
 service for all users in outage.
\label{f.Throughput900bis}}
\begin{center}
\includegraphics[width=0.7\linewidth]{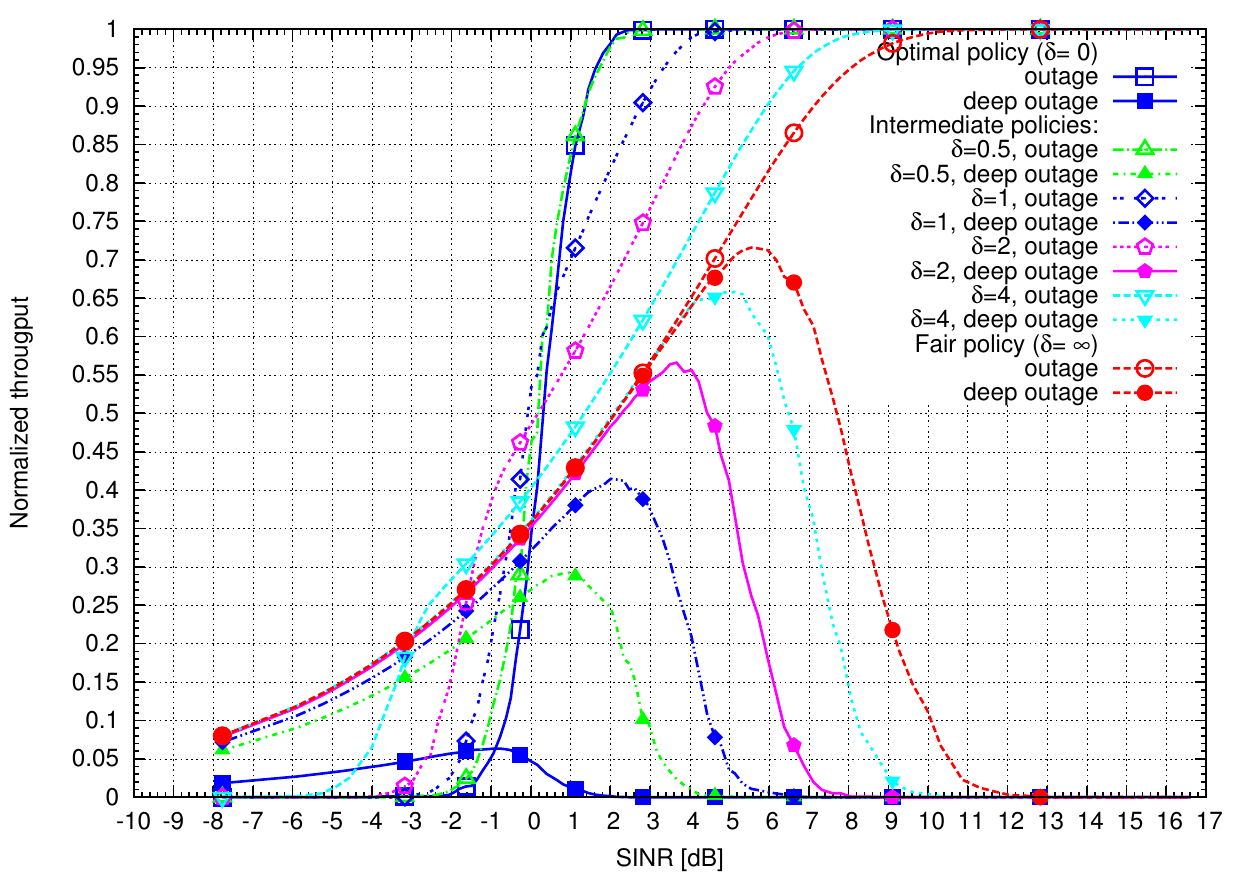}
\end{center}
\vspace{-4ex}
\caption{Mean total throughput normalized to its maximal value
  $256\text{kbit/s}$ obtained  during the service time (upper curves)
and its fraction obtained when a user is in outage (lower curves)
for different policies LESF$(\delta)$  traffic 900~Erlang/km${^2}$. 
\label{f.Throughput900}}
\vspace{-4ex}
\end{figure}

\subsubsection{The role of the ``best effort'' service}
\label{sss.Throughput}
Figure~\ref{f.Throughput900bis} shows the fraction of time spent in 
deep outage  in function of the SINR,
assuming traffic 900~Erlang/km${^2}$. These values should be compared
to the fraction of time spend in outage (for convenience 
copied on Figure~\ref{f.Throughput900bis} from Figure~\ref{f.duration900}).
Recall, users in outage do not receive
the full requested streaming rate (assumed 256kbit/s in our example),
however they do receive some non-null
``best effort''  rates given by~(\ref{e.sub-rate}),
unless they are in deep outage --- have SINR too small;
cf Remark~\ref{r.SINR-outage-region}.
Considering users in outage but not in deep outage 
as ``partially satisfied'', 
increasing fairness margin $\delta$ 
allows to (at least) partially satisfy users with decreasing SINR
values. Obviously the level of the ``partial satisfaction'' depends on the
throughput obtained in outage periods, which is our quantity of
interest on Figure~\ref{f.Throughput900}. 
It shows also two curves for all policies 
LESF$(\delta)$ assuming traffic 900~Erlang/km${^2}$.
The upper ones represent the mean total throughput 
realized during the service, normalized to its maximal value; i.e., 
$T_k/(256\text{kbit/s})$,  
in function of the SINR value characterizing class $k$. 
The fractions of this throughput 
realized during outage periods, $T'_k/(256\text{kbit/s})$, are
represented by the lower curves. 

Figures~\ref{f.Throughput900} and~\ref{f.Throughput900bis} teach us
that the role of the LESF$(\delta)$ policies with $\delta>0$
 may be two-fold.
\begin{itemize}
\item  
 LESF$(\delta)$ policies with
 small values of $\delta$, e.g. $\delta=0.5$, {\em improve 
``temporal homogeneity'' of service with respect to the optimal policy,
for users having SINR near the critical
value.} For example, a user having  SINR equal to $1\text{dB}$ 
is served by the optimal policy during 80\% of time with the full
requested streaming rate (cf. Figure~\ref{f.Throughput900bis}).
However, for the remaining 20\% of time it does not receive any service (deep
outage, rate 0bits/s). 
The policy LESF$(0.5)$ offers to such a user 80\% of the 
requested streaming rate during the whole streaming time
(cf. Figure~\ref{f.Throughput900}),
with no deep outage periods
(cf. Figure~\ref{f.Throughput900bis}). 
The price for this is that 
a slightly higher  SINR is required to receive the full requested
streaming rate (at least  5dB, instead of 3dB for the optimal policy).  
\item The fair policy  LESF$(\infty)$ 
{\em improves the spatial homogeneity of service}. 
It leaves no user in deep outage, however a much larger
SINR$=13$dB is required  for not to 
be in outage (cf. Figure~\ref{f.Throughput900bis}).
Moreover, the throughput of all users 
in outage but not in deep outage 
is substantially reduced 
e.g. from 80\% to 40\% for SINR$=1$dB, 
with respect to some intermediate  LESF$(\delta)$ policies (with
$0<\delta<\infty$). These intermediate policies can offer
an interesting compromise between the optimality and fairness.
 \end{itemize}

\begin{figure}[t!]
\begin{center}
\includegraphics[width=0.7\linewidth]{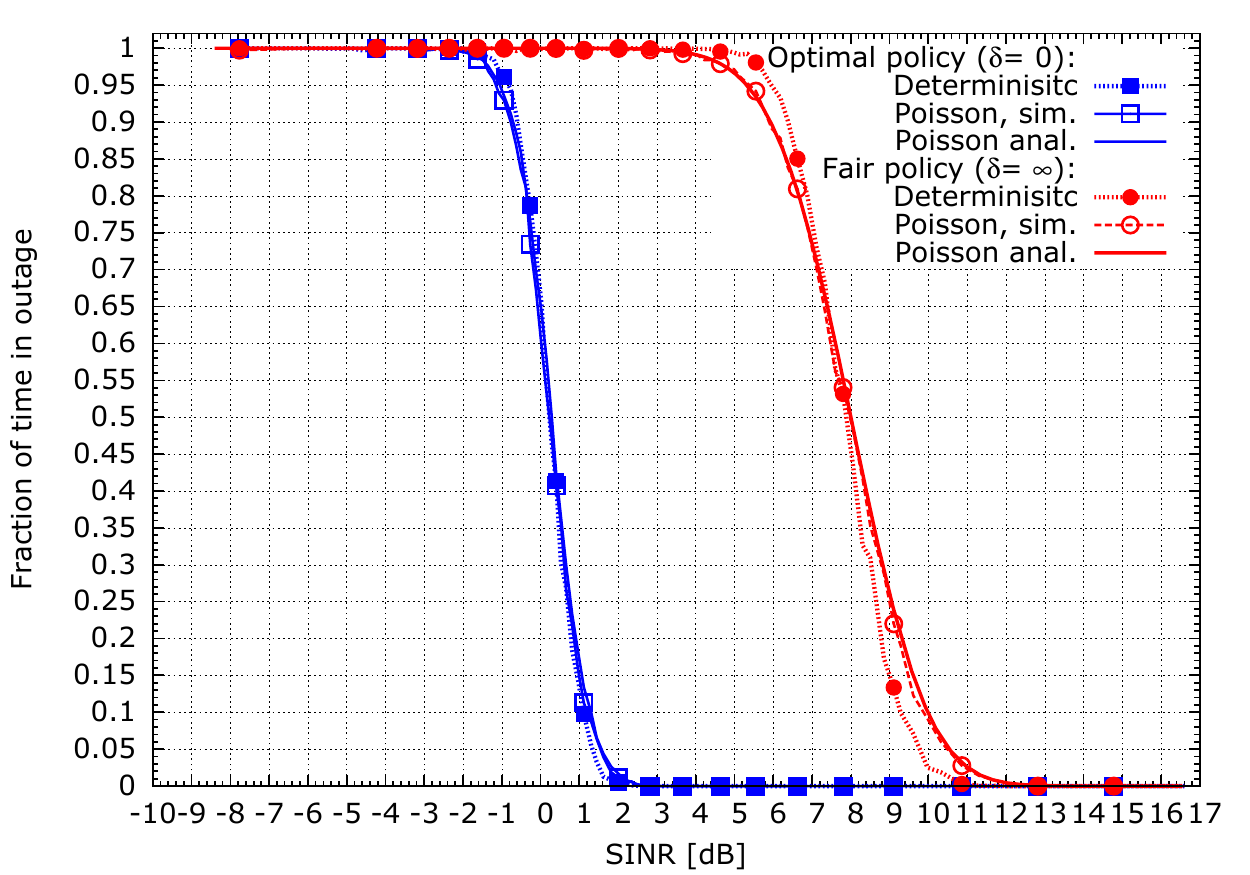}
\end{center}
\vspace{-4ex}
\caption{Impact of the deterministic arrival process (as compared to the Poisson
  one) on the mean 
fraction of the requested streaming time in 
 outage,  for the optimal and fair policy; traffic 900~Erlang/km${^2}$. 
\label{f.det-duration900}}
\begin{center}
\includegraphics[width=0.7\linewidth]{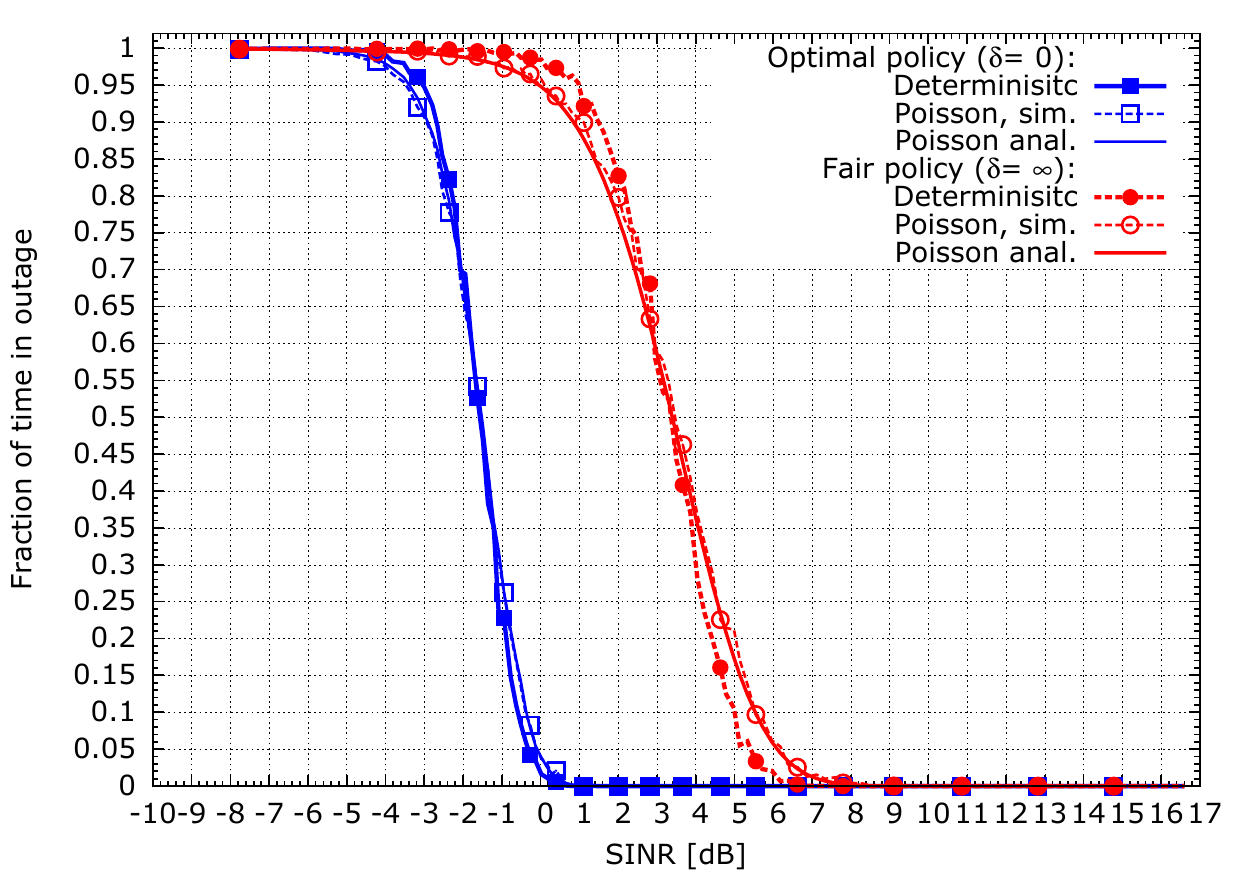}
\end{center}
\vspace{-4ex}
\caption{Impact of the deterministic arrival process (as compared to the Poisson
  one) on the mean 
fraction of the requested streaming time in 
 outage,  for the optimal and fair policy; traffic 600~Erlang/km${^2}$.
\label{f.det-duration600}}
\vspace{-4ex}
\end{figure}

\begin{figure}[t!]
\begin{center}
\includegraphics[width=0.7\linewidth]{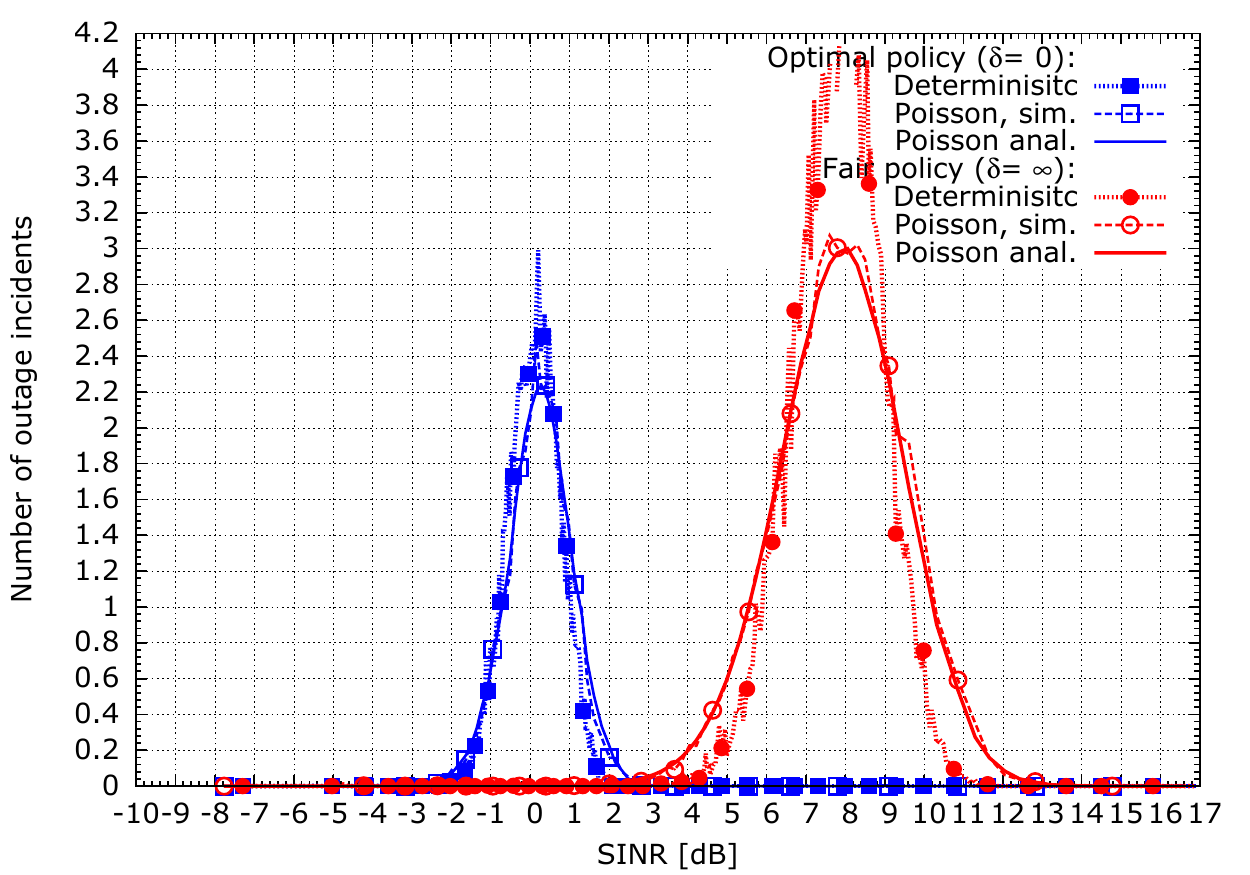}
\end{center}
\vspace{-4ex}
\caption{Impact of the deterministic arrival process (as compared to the Poisson
  one) on the mean number of outage incidents  for the optimal and
  fair policy; traffic 900~Erlang/km${^2}$.
\label{f.det-number900}}
\begin{center}
\includegraphics[width=0.7\linewidth]{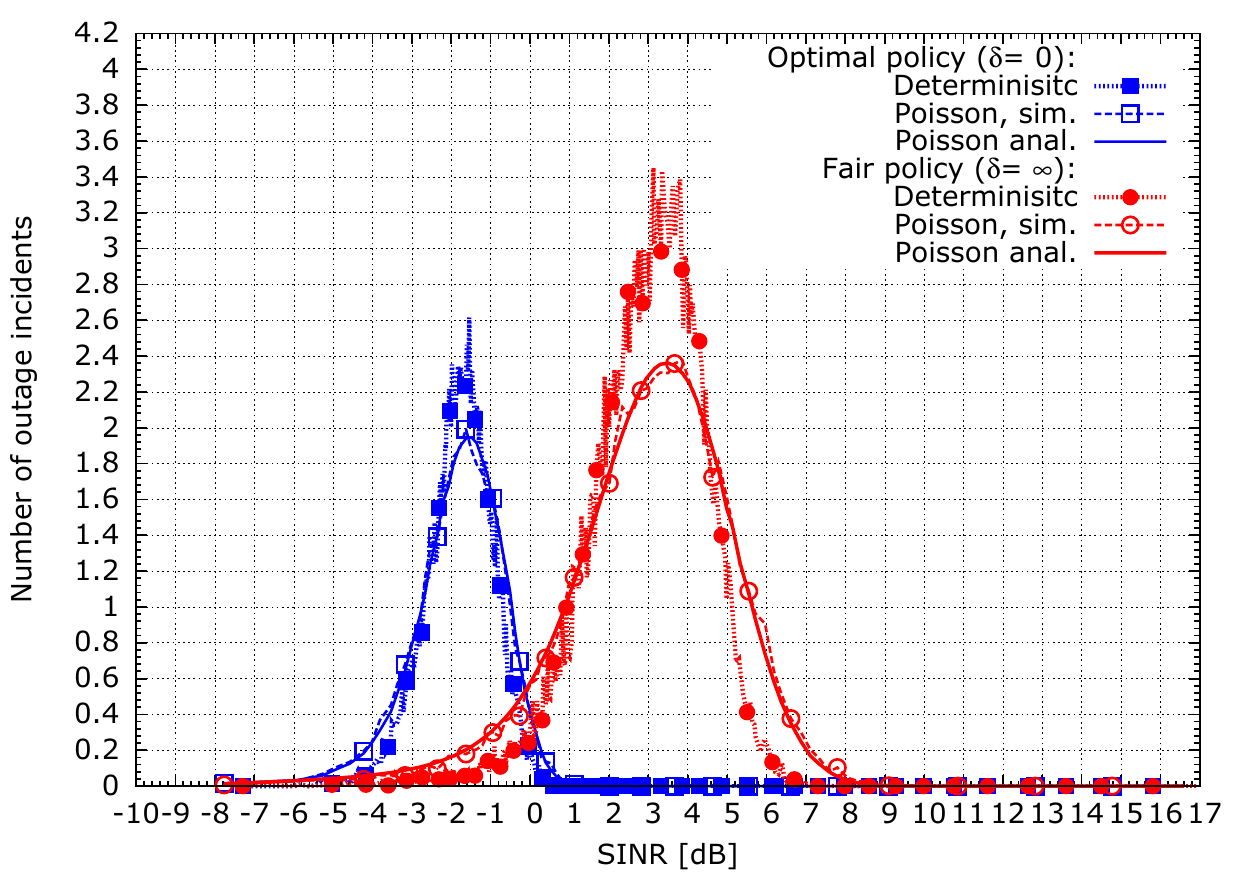}
\end{center}
\vspace{-4ex}
\caption{Impact of the deterministic arrival process (as compared to the Poisson
  one) on the mean number of outage incidents  for the optimal and
  fair policy; traffic 600~Erlang/km${^2}$.
\label{f.det-number600}}
\vspace{-4ex}
\end{figure}

\subsubsection{Impact of a non-Poisson-arrivals}
\label{sss.Deterministic}
Recall that the performance analysis of the  model presented in this paper
is insensitive to distribution of the requested streaming times. 
In this section we will briefly  
study the impact of a non Poisson-arrival assumption. In this regard we 
simulate the dynamics of the model with deterministic inter-arrival times
(with all other model assumptions as before) and estimate the mean 
fraction of time in outage $\mu_kD_k$ and mean number of outage incidents  $M_k$  for each class $k$.
For the comparison, as well as  for the validation of the theoretical work, we 
perform also the simulation of the model with Poisson arrivals. 
The results 
are plotted on Figures~\ref{f.det-duration900}, \ref{f.det-duration600} and 
\ref{f.det-number900}, \ref{f.det-number600}. 
Observe first that the simulations of the Poisson model confirm the results of the 
theoretical analysis. 
Regarding the impact of the deterministic inter-arrival times 
a~(somewhat expected) fact is that 
the optimal policy remains optimal regarding the fraction of time spent in
the outage and the number of outage incidents. 
Another, less evident, observation is that the deterministic
inter-arrivals (more regular than in the Poisson case) 
do {\em not} improve the situation for {\em all} classes of users.
In fact, users with small values of the SINR
have smaller fraction of time in outage under Poisson arrival assumption  
than in the deterministic one! This is different from what we can
observe for the blocking probability for the classical Erlang's loss model; cf
e.g.~\cite[Figure~8]{yu2013g}. Moreover, the
  deterministic arrivals increase the number of outage incidents for
intermediate values of the SINR and decrease for extreme ones, 
especially with the fair policy.
Concluding these observations one can say however, that the  differences
between Poisson and deterministic are not  very significant
and hence we  conjecture that the Poisson model can be used to approximate
a more realistic arrival traffic.

\section{Conclusions}
In this paper, a real-time streaming (RTS) traffic, as e.g. mobile
TV, is analyzed in the context of wireless cellular networks. An
adequate stochastic model is 
proposed to evaluate user  performance metrics, such as
frequency and number of interruptions during RTS  calls
in function of user radio conditions.
Despite some fundamental similarities to the classical
Erlang loss  model, a  new model was required 
for this type of service, where  the service denials 
are not definitive for a given
call, but only temporal -- having the form of, hopefully short, 
interruptions (outage) periods.
Our model  allows to take into account realistic
implementations of the RTS service, e.g. in the LTE  networks.
In this latter context, several numerical demonstrations are given,
presenting the quality of service metrics in  function of user radio
conditions.

\appendix[A general real-time streaming (RTS) model]
\label{s.model}
In this section we will present a general stochastic model for real-time
streaming. An instantiation of this model was used in the main body of the
paper to evaluate the real-time streaming in wireless cellular
networks. 
This model comprises 
Markovian, multi-class process of call arrivals
and their independent, arbitrarily distributed streaming times.
These calls are served by a server whose service capacity is limited.
Depending on numbers of calls of different classes present in the system,
the server may not be able to serve some classes of users.
If such a congestion occurs, these classes
are  temporarily denied the service, until the next call arrival 
or departure, when the situation is reevaluated. 
These service denial periods, called  outage periods, 
do not alter the call sojourn times in the system.
Our model  allows for a very general service (outage)
policy saying which classes of users are temporarily denied the service
due to insufficient service capacity. 
We will evaluate key characteristics of this model 
using the formalism of point processes and their Palm theory, often used in the 
modern approach to stochastic networking~\cite{BaccelliBremaud2003}.
Specifically, we are interested in  the intensity of outage incidents,
the mean inter-outage times  and the outage durations
of a given class, seen from the server perspective, as well as the
probability of outage at the arrival epoch, mean total time in outage
and mean number of outage incidents experienced by a typical user of a
given class. 
The  expressions developed for these characteristics
involve only stationary probabilities of the (free) traffic demand 
process, which in our case is a vector of independent Poisson random
variables. Recall that such a representation is possible e.g. for 
the well known Erlang-B formula, giving the blocking probability in
the classical (possibly multi-class) Erlang's loss model.
Indeed, our model can be seen as an extension of the classical loss
model, where the losses (i.e., service denials) are not definitive
for a given call, but only temporal --- having the form of outage
periods.

\subsection{Traffic demand}

\label{s.UsersLocations}
Consider $J\ge 1$ classes of users identified with calls.
We assume that users of class
$k\in\{1,\ldots,J\}$ arrive in time according
to a Poisson process $N_k=\{T^k_n:n\}$
\footnote{The time instants  $T_n^k$  are used only in the Appendix and should not be 
confused with  $T_k$  denoting in the main stream of the paper (and in the  proof of Proposition~\ref{e.Throughput-Erlang} at the end of the Appendix)
the  mean throughput of user in class $k$.} 
 with intensity $\lambda_{k}>0$ 
and stay in the
system for independent requested streaming 
times $W_n^k$ having some general distribution with mean 
$1/\mu_{k}<\infty$. All the results presented in what follows 
do not depend on the particular choice of the streaming time
distributions --- the property called in the queueing-theoretic context 
{\em insensitivity property}.  
Denote by $\tilde N_k=\{(T^k_n,W^k_n):n\}$ 
the process of arrival epochs  and streaming  times (call durations)
of users of class $k$. We assume that  
$\tilde N_k$ are independent across  $k=1,\ldots,J$.
Denote by $X_k(t)=\sum_n\ind_{[T^k_n,T^k_n+W_n^k)}(t)$
the number of users of class $k$ present in the
system at time~$t$ and let
$\bmX(t)=(X_1(t),\ldots,X_J(t))$; we call it the (vector of) user
configuration at time $t$. The stationary distribution $\pi$ of
$\bmX(t)$ coincides with  
the distribution of a vector
of independent Poisson random variables  $(X_1,\dots,X_J)$
with means $\E[X_k]:=\rho_{k}=\lambda_{k}/\mu_{k}$, 
$k=1,2,\ldots,J$.
We call $\rho_k$ the  \emph{traffic demand} of class $k$.

We adopt the usual convention 
for the numbering of the arrival epochs $T^k_{0}\leq0<T^k_{1}$. The
same convention is used with respect to all  point processes denoting 
some time epochs.

\subsection{Resource constraints and outage policy}
\label{ss.resource-policy}
For class $k=1,\ldots,J$, let a subset of user configurations %
$\calF_k\subset{\bar\bbN}^J$ be given, where
$\bar\bbN=\{0,1,\ldots\}$, 
such that all $X_k$ users of class $k$ present in the configuration 
$\bmX=(X_1\,\ldots,X_k,\ldots,\allowbreak X_J)$  are served if and only if
$\bmX\in\calF_k$ and no user of class $k$ is served (we say  it is in
{\em outage}) if
$\bmX\not\in\calF_k$.   We call $\calF_k$ the {\em $k\,$th class
(service)  feasibility set}.
Denote by $\pi_k=\pi(\calF_k)$ the probability
that the stationary configuration of users is in $k\,$th class
  feasibility set.

We assume that,  upon each
arrival or departure of a user, the system updates its
decision and, for any class $k$, it 
assigns the service to all users of class $k$
if the updated configuration of users is in $\calF_k$. All users 
of any  class $j$ for which the updated configuration is in
$\calF'_k={\bar\bbN}^J\setminus \calF_k$ will be placed in outage (at
least) until the next user arrival or departure. 

In what follows we will assume that 
no user departure can cause outage of any class of users
i.e.,  switch a given configuration from $\calF_k$ to
$\calF'_k$. (However a user departure 
may make  some class $j$ switch from  $\calF'_j$ to  $\calF_j$.)

Denote by $\tilde X_k(t):=X_k(t)\ind_{\calF_k}(\bmX(t))$ the
number of users of class $k$ {\em not in outage} at time~$t$.    
Denote by $\tilde\bmX(t)=(\tilde X_i(t),\ldots,\tilde X_J(t))$ the
configuration of users {\em not} in outage  at time~$t$.

\subsection{Performance metrics}
In what follows we will be interested in the following characteristics
of the model.
\subsubsection{Virtual system metrics}
During its time evolution, the user configuration $\bmX(t)$
alternates visits in the feasibility set $\calF_k$ and its complement
$\calF_k'$, for each class $k=1,\ldots,J$.
We are interested in the expected visit durations in theses sets
as well as the intensities (frequencies) of the alternations. 
More formally, for  each given $k=1,\ldots,J$, we define
the point process $B_k:=\{\tau^k_n:n\}$ of exit  epochs of
$\bmX(t)$ from $\calF_k$; i.e., all epochs $t$ such that 
$\Bigl(\bmX(t-),\bmX(t)\Bigr)\in \calF_k\times\calF_k'$
(with the convention $\tau^k_{0}\leq0<\tau^k_{1}$).
These are epochs when all users of class $k$ present in the system 
(if any)  have their service interrupted.

Denote by $\sigma'^k_n:=\sup\{t-\tau_n^k:\bmX(s)\in\calF'_k\; \forall
s\in[\tau_n^k,t)\}$ the duration of the $n\;$th visit of the
process $\bmX(t)$  in $\calF'_k$
and  by $\sigma^k_n:=\tau^k_{n+1}-\tau^k_n-\sigma'^k_n$ the duration of
the $n\;$th visit of the process $\bmX(t)$  in $\calF_k$.
We define for each class $k=1,\ldots,J$:
\begin{itemize}
\item {\em The intensity of outage incidents 
of class $k$}, i.e., the mean 
number of outage incidents of this class per unit of time
$$\Lambda_k:=\lim_{T\to\infty}\frac1T
\sum_n\ind_{[0,T)}(\tau_n^k)\,.$$
Obviously $\Lambda_k$ is also the intensity of entrance 
to the  $k\,$th class feasibility set  $\calF_k$.
\item {\em The mean service time between two outage incidents
of class  $k$}
$$\bar\sigma_k:=\lim_{N\to\infty}\frac1N
\sum_{n=1}^N\sigma_n^k\,.$$
\item {\em The mean outage duration of class  $k$}
$$\bar\sigma'_k:=\lim_{N\to\infty}\frac1N 
\sum_{n=1}^N\sigma'^k_n\,.$$
\end{itemize}
Note that the above metrics characterize a ``virtual'' quality of the 
service, since some visits in  $\calF_k$ and  $\calF'_k$
may occur  when there is no $k\,$th class user in the system
(in the latter case the outage of this class is not experienced by any user).

\subsubsection{User metrics}
We adopt now a user point of view on the system. We define for each
class $k=1,\ldots,J$:
\begin{itemize}
\item  {\em The probability of outage at the arrival epoch for  user
of class~$k$}
$$P_k=\lim_{N\rightarrow\infty}\frac1N
\sum_{n=1}^N\ind_{\calF'_k}(\bmX(T^k_n))\,.$$ 
\item  {\em The mean total time in outage of user of class $k$}
$$D_k=\lim_{N\rightarrow\infty}\frac1N
\sum_{n=1}^N\int_{[T^k_n,T^k_n+W^k_n)}\ind_{\calF'_k}(\bmX(t))\,\md t\,.$$ 
\item {\em The mean  number of outage incidents experienced by user of
    class $k$ after its arrival}
$$M_k=\lim_{N\rightarrow\infty}\frac1N
\sum_{n=1}^N\sum_{m}\ind_{(T^k_n,T^k_n+W^k_n)}(\tau^k_m)\,.$$
Note that eventual outage experienced  at the 
arrival of a given user is not counted in~$M_k$. 
The mean  total 
number of outage incidents (including possibly at the arrival
epoch) is hence  $P_k+M_k$. 
\end{itemize}

\section{Mathematical results}
\label{s.results}
For a given class $k=1,\ldots,J$, denote by
$\varepsilon_k=(0,\ldots,1,\ldots,0)\in\bar\bbN^J$ the unit vector
having its $k\,$th 
component equal to~1. Hence 
$\bmx+\varepsilon_k$ represents adding one user of class $k$ 
to the configuration
of users $\bmx\in\bar\bbN^J$. Denote by $\Pr$ 
the probability under which  $\{\bmX(t):t\}$ is stationary and by $\E$ the
corresponding expectation.
Recall that $\pi\{\bmx\in\cdot\} =\Pr\{\bmX(t)\in\cdot\}$ is the distribution
of the stationary configuration of users $\bmX(t)$ (it corresponds
to independent Poisson variables of mean $\rho_k$). 

\subsection{General results}
\label{ss.general-results}
We present first results regarding the virtual system metrics. These
results will be next used to evaluate the user metrics.

\begin{Lemma}
\label{l.Intensity}\label{l.Lambda}
The intensity of outage incidents 
of class $k$ 
is $\Pr$-almost surely equal to
$$\Lambda_{k}
=\sum_{j=1}^{J}\lambda_{j}\pi\left\{\bmx\in\calF_k,\bmx+\varepsilon_{j}%
\in\calF'_k\right\}\,\qquad k=1,\ldots,J.$$
\end{Lemma}

\begin{proof}
Let $N=\sum_{j=1}^JN_j$\ be the point process\ counting the 
arrival  times of users of all classes. By independence, $N$  is the Poisson
point process of intensity $\lambda=\sum_{j=1}^J\lambda_j$. 
Then, by the ergodicity of the process $\{\bmX(t):t\}$
and the fact that the exits from $\calF_k$ can take place only at some
user arrival epoch we have by the Campbell's formula
\citep[cf. e.g.][Equation~(1.2.19)\footnote{with
  $Z_n:=(\bmX(T_n-),\bmX(T_n))$ and
  $f(t,z)=\ind_{[0,1)}(t)\ind_{\calF_k\times\calF'_k}(z)$}]%
  {BaccelliBremaud2003},
$$
\Lambda_{k} 
=\E\left[  \int_{[0,1)}
\ind_{\calF_k\times\calF'_k}\left(
\bmX\left(  t-\right)  ,\bmX\left(  t\right)  \right)  N\left(  \md
t\right)  \right] =\lambda
\Pr_{N}^0\left\{\bmX(0-)\in\calF_k,\bmX(0)\in\calF'_k\right\}
\,,$$
where $\Pr_{N}^{0}$ designates the Palm
probability associated to $N$ (which is, roughly speaking, the
conditional probability given  an arrival at time $0$).
By  PASTA (Poisson Arrivals See
Time Averages) property~\citep[cf.][Equation~(3.3.4)]{BaccelliBremaud2003}
the configuration of users $\bmX(0-)$ under  $\Pr_{N}^{0}$
has distribution $\pi$. Moreover, $\bmX(0)=\bmX(0-)+\varepsilon_{\xi}$
where $\xi\in\{1,\ldots,J\}$ is under  $\Pr_{N}^{0}$ independent of
$X(0-)$ and takes value $j$ with probability $\lambda_j/\lambda$.
This completes the proof.
\end{proof}

\begin{Lemma}
The mean service time between two outage incidents and 
the mean outage duration of class  $k$
are  $\Pr$-almost surely  equal to,  respectively,
\[
\bar\sigma_k=\frac{\pi\left(  \calF_k\right)  }{\Lambda_{k}},\quad
\bar\sigma'_k:=\frac{\pi(  \calF'_k)  }%
{\Lambda_{k}}%
\,\quad \qquad k=1,\ldots,J,\]
where $\Lambda_{k}$\ is given in Lemma~\ref{l.Intensity}.
\end{Lemma}

\begin{proof}
First we prove the expression for $\bar\sigma_k$.
By ergodicity 
$\bar\sigma_k=\E_{B_k}^{0}\left[  \sigma^k_{0}\right]$  $\Pr$-almost surely,
where $\E_{B_k}^{0}$\ designates the expectation with respect to the Palm
probability associated to $B_k$,
and 
$\E_{B_k}^{0}\left[  \tau^k_{0}\right]  =1/\Lambda_{k}$;
\citep[see e.g.][Equation~(1.6.8) and Equation~(1.2.27)]{BaccelliBremaud2003}.
Applying the mean value formula~\citep[see][Equation~(1.3.2)\footnote{
with $Z_k\left(  t\right)  =\ind_{\calF_k}\left(\bmX\left(  t\right)
\right)$}]{BaccelliBremaud2003}
 we get 
$\pi(\calF_k)
=\Lambda_{k}\E_{B_k}^{0}\left[ \sigma^k_{0}\right]$,
which completes the proof of the expression for $\bar\sigma_k$.
For the other expression, note by the definition of the sequence
$\sigma^k_n,\sigma'^k_n$ and $\tau^k_n$ that $\Pr$-almost surely,
$$\bar\sigma'_k=\E_{B_k}^{0}\left[  \sigma'^k_{0}\right]=
\E_{B_k}^{0}\left[\tau^k_1-\sigma^k_{0}\right]=
\frac1{\Lambda_k}-\frac{\pi(\calF_k)}{\Lambda_k}=
\frac{\pi(\calF'_k)}{\Lambda_k}\,,$$
which completes the proof.
\end{proof}

\begin{Prop}\label{p.InterruptionProbability}
The probability of outage at the arrival epoch for  user
of class~$k$
is equal to 
\begin{equation}
P_k=\pi\left\{ \bmx+\varepsilon_{k}\in
\calF'_k\right\} \qquad k=1,\ldots,J\,  \label{e.InterruptionProbability}%
\end{equation}
$\Pr$-almost surely.
\end{Prop}

\begin{proof}
By ergodicity we  have 
$P_k=\Pr_{N_k}^0\left\{\bmX(0)\in\calF'_k\right\}$,
where  $\Pr_{N_k}^{0}$ designates the Palm
probability associated to $N_k$ (arrival process of the users of class
$k$). By  PASTA property
the configuration of users $\bmX(0-)$,
just before arrival of the user of class $k$ at time $0$,
has distribution $\pi$. Once the user
enters the system, the users configuration becomes
$\bmX(0-)+\varepsilon_{k}$, whence the result.
\end{proof}

\begin{Prop}\label{p.InterruptionDuration}%
The mean total time in outage of user of class $k$
is $\Pr$-almost surely equal to 
$$
D_k=\frac{1}{\mu_{k}}\pi\left\{\bmx+\varepsilon
_{k}\in\calF'_k\right\}\, \qquad k=1,\ldots,J
\,. 
$$
\end{Prop}

\begin{proof}
Again using the ergodicity of $\left\{  \bmX\left(  t\right)  \right\}
$  we can write
$$D_k=\E_{N_k}^0\left[\int_{[0,W^k_0)}\ind_{\calF'_k}(\bmX(t))
\,\md t\right]\,.$$ 
Denote by $\bmY(t):=\bmX(t)-\varepsilon_k\ind_{[T^k_0,T^k_0+W^k_0)}(t)$ 
the process of configurations of users other than the user number 0 of
class $k$ (which arrives at time~$0$ under $\E_{N_k}^0$).
By Slivnyak theorem~\citep[see e.g.][Theorem~1.13]{BaccelliBlaszczyszyn2009T1}
the distribution of the process $\{\bmY(t):t\}$ under $\Pr_{N_k}^0$ is the same 
as this of $\{\bmX(t): t\}$ under $\Pr$. 
Using the fact that  $W^k_0$ and $\bmY(t)$ are independent under
$\Pr_{N_k}^0$ with $\E_{N_k}^0[W^k_0]=1/\mu_k$ we obtain
$$D_k=
\int_0^\infty \E_{N_k}^0\left[\ind_{[0,W^k_0)}(t)
\ind_{\calF'_k}(\bmY(t)+\varepsilon_k)\right]\,\md t=
\frac 1\mu_k\pi\left\{\bmx+\varepsilon_k\in\calF'_k)\right]\,,
$$
which completes the proof. 
\end{proof}

\begin{Prop}
\label{p.NbInterruptions}
 The mean  number of outage incidents experienced by user of
    class $k$ after its arrival is
$\Pr$-almost surely equal to 
\begin{equation}
M_k=\frac{1}{\mu_{k}}\sum_{j=1}^{J}\lambda
_{j}\pi\left\{\bmx+\varepsilon_{k}\in\calF_k,\bmx+\varepsilon_{k}%
+\varepsilon_{j}\in\calF'_k\right\}, 
\qquad k=1,\ldots,J\,.  \label{e.NbInterruptions}%
\end{equation}
\end{Prop}

\begin{proof}
Again using the ergodicity of $\left\{  X\left(  t\right)  \right\}  $ we
know  that, $\Pr$-almost surely,
\[M_k=
\E_{N_k}^{0}\left[\int_{(0,W^k_0)}\,B_k(\md t)\right]\,.
\]
Using the fact that  $W^k_0$ and $\bmY(t)$ are independent under
$\Pr_{N_k}^0$ with $\E_{N_k}^0[W^k_0]=1/\mu_k$ we obtain
$$M_k=\E_{N_k}^0\left[B^*_k(0,W^k_0)\right]=\frac{\Lambda^*_k}{\mu_k}$$
 where  
$B^*_k=:\{\tau^{*k}_n:n\}$ is the point process of  exit  epochs of
$\bmX(t)$ from $\calF^*_k=\{\bmx:\bmx+\varepsilon_k\in\calF_k\}$
and $\Lambda^*_k$ its intensity. Using Lemma~\ref{l.Lambda} 
with $\calF_k$ replaced by $\calF^*_k$ concludes the proof.
\end{proof}

We will now prove the result regarding the throughput of the typical
call of class~$k$.
\begin{proof}[Proof of Proposition~\ref{e.Throughput-Erlang}]
We have
$$T_k=T_k^{\delta}=\mu_k\E^0_{N_k}\left[\int_{[0,W_0^k)}
    r_k\ind(\bmX(t)\in\calF_k^\delta)
+r_k^{'\delta}(\bmX(t))\ind(\bmX(t)\not\in\calF_k^\delta)\,\md t\right]\,.$$
It is easy to see, as in the proof of  Proposition~\ref{p.InterruptionDuration},
that 
$T_k=r_k\pi\left\{\bmx+\varepsilon_{k}\in\calF_k^\delta\right\}+T_k'$, where
\begin{equation}
T_k'=\E\left[r_k^{'\delta}(\bmX(t)+\varepsilon_{k})\ind((\bmX(t)+\varepsilon_{k})\not\in\calF_k^\delta)\right]
\end{equation}
is the part of the throughput  obtained by 
user of class $k$ during its outage time.
\end{proof}

\bibliographystyle{IEEEtranTCOM}
\addcontentsline{toc}{section}{References}

\end{document}